\documentclass[12pt]{article}
\usepackage[verbose,a4paper,tmargin=0.5in,lmargin=1in,rmargin=1in]{geometry}
\usepackage[pdftex]{graphicx}
\usepackage{amsfonts}
\usepackage{indentfirst}
\usepackage{latexsym}
\usepackage{amsmath}
\usepackage{amsthm}
\usepackage{amssymb}
\usepackage{psfrag}
\usepackage[mathscr]{eucal} 
\usepackage{color}
\usepackage{cite}
\usepackage{longtable}
\usepackage{verbatim}
\usepackage[normalem]{ulem}
\usepackage{tensor}
\usepackage{appendix}

\usepackage{graphicx}
\usepackage[all]{xy}
\usepackage{float}

\usepackage{physics}

\newtheorem{Twierdzenie}{Theorem}[section]
\newtheorem{Definicja}{Definition}[section]
\newtheorem{Lemat}{Lemma}[section]
\newtheorem{Uwaga}{Remark}[section]
\newtheorem{Wniosek}{Corollary}[section]

\newfloat{Scheme}{tbp}{muj}

\newcommand{\zmianamd}[1]{\textcolor{blue}{#1}}

\title{An example of algebraically general para-K\"ahler Einstein space}

\author{$\textrm{Adam Chudecki}^{a,b,c}$, \ \ $\textrm{Micha\l{} Dobrski}^{b,d}$}

\begin{document}

\maketitle

$^a$ Center of Mathematics and Physics, Lodz University of Technology, Al. Politechniki 11, 90-924 Lodz, Poland. 
\newline
$^b$ Institute of Physics, Faculty of Technical Physics, Information Technology and Applied Mathematics, Lodz University of Technology, Wólcza\'nska 217/221, 93-005 Lodz, Poland.
\newline
$^c$ adam.chudecki@p.lodz.pl
\newline
$^d$ michal.dobrski@p.lodz.pl
\newline
\newline
\newline
\textbf{Abstract}. 
\newline
Algebraically general para-K\"ahler Einstein spaces equipped with 3D algebras of infinitesimal symmetries are considered. It is shown that if the algebra contains 2D trivial subalgebra then vacuum Einstein field equations with cosmological constant can be reduced to a single, first-order differential equation. One of the cases is solved explicitly. Hence, a first example of an algebraically general para-K\"ahler Einstein space is given.
\newline
\newline
\textbf{Key words:} para-K\"ahler Einstein spaces, algebraically general exact solutions, Killing vectors.


\section{Introduction}
\setcounter{equation}{0}

The present paper is motivated by recent advances in understanding para-K\"ahler geometry and its role in mathematical physics \cite{Bor_Lamoneda_Nurowski, Bor_Makhmali_Nurowski, Nurowski_An}. At the beginning of the previous decade para-K\"ahler structures, as  structures equipped with a real metric of neutral signature $(++--)$ definitely were not in the research mainstream of theoretical physics. At best they were viewed as a territory located in the border area between physics and mathematics. Then An and Nurowski \cite{Nurowski_An} proved that it is not entirely accurate picture and para-K\"ahler Einstein structures play remarkable role in mathematical physics as they naturally appear in its noble core -- classical theoretical mechanics.

Let us recall the concept of \textsl{$(2,3,5)$-distribution} in dimension 5. Let $\mathcal{Q}$ be a 5D manifold and $\mathcal{D} \subset T \mathcal{Q}$ is a distribution such that $\textrm{rank} ( \mathcal{D}) = 2$. If we denote the n-th derived system of $\mathcal{D}$ by $\partial^{n} \mathcal{D}$ then the distribution $\mathcal{D}$ is a $(2,3,5)$-distribution if $\textrm{rank} (\partial \mathcal{D}) = 3$ and $\textrm{rank} (\partial^{2} \mathcal{D}) = 5$. The Lie algebra of infinitesimal symmetries of a $(2,3,5)$-distribution is at most 14 dimensional and in such a case it is isomorphic to the split real form of the simple exceptional Lie algebra $\mathfrak{g}_{2}$. Cartan in \cite{Cartan} solved the equivalence problem for $(2,3,5)$-distributions and found a set of local differential invariants for such distributions. One of these invariants, called \textsl{a Cartan quartic} plays an especially important role in the problem of equivalence and classification of $(2,3,5)$-distributions.

But $(2,3,5)$-distributions are interesting not only from a mathematical point of view. There are examples of physical systems which are 5 dimensional and which are naturally equipped with a $(2,3,5)$-distribution. One is remarkable: it is the configuration space of a system of two rigid bodies which roll on each other. Let two rigid bodies be bounded by two surfaces $\Sigma_{1}$ and $\Sigma_{2}$ equipped with metrics of Riemannian signature $(++)$, $g_{1}$ and  $g_{2}$, respectively. The configuration space $\mathcal{C} (\Sigma_{1}, \Sigma_{2})$ is a 5D circle bundle. It turns out that, if there is no slipping or twisting in the relative motion of rolling bodies, then $\mathcal{C} (\Sigma_{1}, \Sigma_{2})$ is equipped with rank 2 distribution which is called \textsl{the rolling distribution}. The problem: find pairs of surfaces $(\Sigma_{1}, g_{1})$ and $(\Sigma_{2}, g_{2})$ for which the symmetry of the rolling distribution is $\mathfrak{g}_{2}$ has been solved in the paper \cite{Nurowski_An}. It was proved later \cite{Bor_Lamoneda_Nurowski} that $(2,3,5)$-distributions with maximal algebra of infinitesimal symmetries can be obtained from the homogeneous para-K\"ahler Einstein metric which is called \textsl{the dancing metric}. 

The rolling problem has also its twistorial interpretation. Indeed, it is enough to define a 4D manifold $\mathcal{M} = \Sigma_{1} \times \Sigma_{2}$ and equip it with a neutral signature metric $g = g_{1} \oplus (-g_{2})$. Then \textsl{a circle twistor bundle} $\mathbb{T} (\Sigma_{1}, \Sigma_{2})$ is isomorphic to $\mathcal{C} (\Sigma_{1}, \Sigma_{2})$. This isomorphism defines a rank 2 distribution on the circle twistor bundle which is called \textsl{a twistor distribution}. 

However, a twistor circle bundle can be defined over an arbitrary 4D manifold (not necessarily $\Sigma_{1} \times \Sigma_{2}$) equipped with a neutral signature metric. Such a neutral structure naturally induces a twistor distribution on $\mathbb{T} (\mathcal{M})$ which is $(2,3,5)$ if the self-dual (SD) Weyl tensor of $\mathcal{M}$ does not vanish. An open question was whether the Cartan quartic of twistor distributions can be of an arbitrary Petrov-Penrose type? It turns out that the Cartan quartic of a twistor distribution and the anti-self-dual (ASD) Weyl tensor of an underlying para-K\"ahler Einstein structure has the same Petrov-Penrose type \cite{Bor_Makhmali_Nurowski}. Hence, the question can be reformulated in the following way: 
\begin{equation}
\nonumber
\textbf{Q}: \ \ \ \textrm{``Can a para-K\"ahler-Einstein metric be of an arbitrary Petrov-Penrose type?''}
\end{equation}

In the present paper we complete the task of answering this question. Let us sketch the geometric path we have chosen to deal with it.

\subsection{Role of the second Pleba\'nski formalism}

It is relatively easy to find a general form of a pKE-metric\footnote{In the paper we use the abbreviation \textsl{pKE-space} (\textsl{pKE-metric}) for \textsl{a para-K\"ahler Einstein space} (\textsl{a para-K\"ahler Einstein metric}). Also, \textsl{an algebraically general para-K\"ahler Einstein space} (\textsl{an algebraically general para-K\"ahler Einstein metric}) is abbreviated by \textsl{gpKE-space} (\textsl{gpKE-metric}).} \cite{Przanowski_Formanski_Chudecki, Bor_Makhmali_Nurowski}. Because pKE-spaces are equipped with two families of totally null planes, the construction of the metric is quite straightforward and it yields to the following, well-known result
\begin{equation}
\label{para-Kahler_pierwszy_formalizm_metryka}
\frac{1}{2} ds^{2} = V_{w \widetilde{w}} \, dw d \widetilde{w} + V_{w \widetilde{z}} \, dw d \widetilde{z} + V_{z \widetilde{w}} \, dz d \widetilde{w} + V_{z \widetilde{z}} \, dz d \widetilde{z} = e^{1} e^{2} + e^{3} e^{4}
\end{equation}
where $(w, \widetilde{w}, z, \widetilde{z} )$ are local real coordinates, $V = V (w, \widetilde{w}, z, \widetilde{z} )$ is a smooth function and $e^{i}$, $i=1,2,3,4$ are the members of a null tetrad
\begin{eqnarray}
\label{null_tetrad_first_formalism}
&& e^{1} = dw, \ e^{2} = V_{w \widetilde{w}} \, d \widetilde{w} + V_{w \widetilde{z}} \, d \widetilde{z}
\\ \nonumber
&& e^{3} = dz, \ e^{4} = V_{z \widetilde{w}} \, d \widetilde{w} + V_{z \widetilde{z}} \, d \widetilde{z}
\end{eqnarray}
Vacuum Einstein field equations with cosmological constant $\Lambda$ are reduced to the equation
\begin{equation}
\label{para-Kahler_pierwszy_formalizm_rownanie}
V_{w \widetilde{w}} V_{z \widetilde{z}} - V_{w \widetilde{z}} V_{z \widetilde{w}} = e^{- \Lambda V}
\end{equation}

The metric (\ref{para-Kahler_pierwszy_formalizm_metryka}) and Eq. (\ref{para-Kahler_pierwszy_formalizm_rownanie}) both remain valid, if one considers a complex generalization: a 4D complex manifold equipped with a holomorphic metric. In such a case the coordinates $(w, \widetilde{w}, z, \widetilde{z} )$ are complex and the function $V$ is holomorphic. Such spaces are  called \textsl{complex para-K\"ahler Einstein spaces} \cite{Przanowski_Formanski_Chudecki}.

The SD Weyl tensor of the metric (\ref{para-Kahler_pierwszy_formalizm_metryka}) is of type [D] if $\Lambda \ne0$ and it is of type [O] if $\Lambda = 0$. In the latter case we arrive at \textsl{the heavenly metric} and Eq. (\ref{para-Kahler_pierwszy_formalizm_rownanie}) becomes the famous \textsl{heavenly equation}. There is one more advantage of a complex approach: the simplicity of obtaining real slices of a complex metric (\ref{para-Kahler_pierwszy_formalizm_metryka}). If we treat the coordinates $(w, \widetilde{w}, z, \widetilde{z} )$ as real ones, we obtain a metric of neutral signature $(++--)$ while if we consider the coordinates as complex such that $\widetilde{w} = \bar{w}$ and $\widetilde{z} = \bar{z}$ where bar denotes a complex conjugation, we arrive at the Riemannian slice with the signature $(++++)$. 

It may seem that the metric (\ref{para-Kahler_pierwszy_formalizm_metryka}) and Eq. (\ref{para-Kahler_pierwszy_formalizm_rownanie}) are perfectly designed for a problem of finding exact solutions of pKE-spaces. After all, in this approach the null tetrad (\ref{null_tetrad_first_formalism}) is exactly adapted to the structure of null planes in the sense that the first family of such planes are given by $w=\textrm{const}$, $z=\textrm{const}$ and the second one by $\widetilde{w}=\textrm{const}$, $\widetilde{z}=\textrm{const}$. 

However one quickly runs into serious computational difficulties when following this path.
Indeed, formulas for the Ricci rotation coefficient and ASD Weyl curvature coefficients $\dot{C}^{(i)}$ are complicated (compare Section 4.1 of \cite{Chudecki_Examples} where these formulas are collected). It is unclear how to impose the condition of algebraic degeneracy of the ASD Weyl tensor, let alone the case of algebraically general ASD Weyl tensor. Fortunately, there is another way.

This way is called \textsl{the second Pleba\'nski formalism}. It allows to describe slightly more general spaces then these for which the metric has the form (\ref{para-Kahler_pierwszy_formalizm_metryka}). Remember, that complex pKE-spaces are equipped with two families of SD null planes (hence, the SD Weyl tensor must be of type [D] or [O]). If one assumes the existence of only one such family, then algebraic degeneracy is the only resulting restriction on SD Weyl tensor (this follows from the \textsl{complex Goldberg-Sachs theorem} \cite{Plebanski_Hacyan}). A space with such a property is called \textsl{the hyperheavenly space ($\mathcal{HH}$-space)}. Formulas for the ASD curvature coefficients are much more convenient (compare (\ref{wspolczynniki_ASD})). Fortunately, the $\mathcal{HH}$-space can be easily equipped with the second family of SD null planes what reduces the SD Weyl tensor to be of the type [D] or [O] and gives an alternative description of complex pKE-spaces. The metric has form (\ref{metryka_w_xyqp}) and vacuum Einstein field equations are reduced to Eq. (\ref{rownanie_hiperniebianskie}). Real neutral slices of the metric (\ref{metryka_w_xyqp}) are equally easily obtained as the neutral slices of the metric (\ref{para-Kahler_pierwszy_formalizm_metryka}). Treating coordinates $(q,p,x,y)$ as real ones and the function $\Theta = \Theta (q,p,x,y)$ as a smooth one, we get real pKE-space. Also, the SD null planes (which are complex surfaces in a complex case) become real ones. 

The second Pleba\'nski formalism allowed to find examples of all algebraically degenerate pKE-spaces \cite{Chudecki_Examples}. A significant generalization were introduced in \cite{Chudecki_Ref_1, Chudecki_Ref_2} where algebraically degenerate pKE-metrics with the ASD Weyl tensor of types [II], [III] and [N] were found in all the generality. (It should be mentioned, that  Nurowski had examples of all algebraically degenerate pKE-spaces already in 2015. These examples were later re-examined and published in \cite{Bor_Makhmali_Nurowski}). A special attention should be paid on pKE-spaces with the ASD Weyl tensor of the type [D] which have been found in all the generality in \cite{Bor_Makhmali_Nurowski}. Hence, all algebraically degenerate pKE-metrics are known and they are completely general.

Results of \cite{Bor_Makhmali_Nurowski, Chudecki_Ref_1, Chudecki_Ref_2} gave a positive answer to the question $\textbf{Q}$ in almost all cases, except one: an algebraically general pKE-space. In this paper we fill this gap.

It is known that algebraically general, vacuum solutions to Einstein field equations are quite problematic. Within the general theory of relativity there are only a few known solutions which belong to this class. It is \textsl{the Petrov solution} introduced in \cite{Petrov} which is the only vacuum solution of
Einstein equations admitting a simply-transitive four-dimensional maximal group of motions. Also \textsl{the Kasner metric} and \textsl{the Taub metric} can be algebraically general.

There is no reason to suspect that algebraically general neutral spaces are less complicated then their Lorentzian counterparts, even if the considerations are restricted to para-K\"ahler class. It quickly becomes apparent that our main goal - to find explicit example of gpKE-space - is an ambitious task, which can only be achieved by some methodical approach. Therefore, the other goal of our work emerges: an analysis of gpKE-spaces with infinitesimal symmetries. A systematic study brings us to the conclusion that manageable yet nondegenerate cases correspond to 3D algebras of Killing vectors which contain a trivial 2D subalgebra. We prove that for all gpKE-spaces equipped with such an algebra the vacuum field equations can be reduced to a first-order ODE known as \textsl{an Abel equation of the second kind}. Finally, as a realization of the main aim of our paper, we present the explicit metric (\ref{jedyna_jawna_metryka}).

\subsection{Structure of the paper}

The paper is organized as follows. In Section \ref{Sekcja_2} we recall the definition of pKE-spaces and $\mathcal{HH}$-spaces such that SD Weyl tensor is of type [D]. We discuss the Killing symmetries in $\mathcal{HH}$-spaces and we list criteria for the ASD Weyl tensor to be algebraically general. Section \ref{Section_3} is devoted to gpKE-spaces with symmetries. We prove that if one assumes the existence of two commuting Killing vectors, then they can be always brought to the form $\partial_{q}$ and $\partial_{p}$ without any loss of generality. Even with such a strong symmetry assumed, the $\mathcal{HH}$-equation is still complicated, compare Eq.\  (\ref{rownanie_hiperniebianskie_dwa_komutujace_Killingi}). Thus, we equip a space with one more Killing vector. Hence, we arrive at 3D algebras of infinitesimal symmetries. We also recall how to classify such algebras (we follow the work by Patera, Sharp and Winternitz \cite{Patera}). 

In Section \ref{Sekcja_4} we analyze 3D algebras of infinitesimal symmetries which contain $2A_{1}$ as an subalgebra. For all the cases we manage to reduce the $\mathcal{HH}$-equation to \textsl{an Abel equation of the second kind}\footnote{Let us remind that a first--order ODE is called \textsl{an Abel equation of the second kind} if it has the form $u u_{x}=  f (x) \, u^{2} + g (x) \, u + h (x)$ where $u=u(x)$ is an unknown function while $f(x)$, $g(x)$ and $h(x)$ are arbitrary functions. Abel equation of the second kind can be always brought to its \textsl{canonical form} which reads $\tilde{u} \tilde{u}_{\tilde{x}} - \tilde{u} =  \Psi (\tilde{x})$.}. Unfortunately, in all the cases, except one, the resulting equation is complicated and its general solution is not known. Finally, a special attention is paid to the case $A_{3,5}^{-\frac{1}{2}}$ (Section \ref{Sekcja_5}). Within this case we are able to find an explicit example. It is the metric (\ref{jedyna_jawna_metryka}). Some concluding remarks end the paper.

All considerations are local and they remain valid in a complex case (coordinates are complex and functions are holomorphic) as well as in a real neutral case (coordinates are real and functions are smooth). However, to analyze 3D algebras of infinitesimal symmetries we use the classification from \cite{Patera} which is valid for real Lie algebras. It is known that complex transformations of algebra generators can make certain types of algebras listed in \cite{Patera} isomorphic.


\renewcommand{\arraystretch}{1.5}
\setlength\arraycolsep{2pt}
\setcounter{equation}{0}

\section{Algebraically general para-K\"ahler Einstein spaces}
\label{Sekcja_2}
\setcounter{equation}{0}

In this section we recall crucial concepts concerning para-K\"ahler geometry and its relation to hyperheavenly spaces.

\subsection{Para-K\"ahler Einstein spaces}

A para-K\"ahler space is defined as follows
\begin{Definicja}
\label{Definicja_para_Kahler}
A para-K\"ahler space is a triple $(\mathcal{M}, ds^{2}, K)$ where $\mathcal{M}$ is a 4-dimensional real smooth manifold, $ds^{2}$ is a real smooth metric of the neutral signature $(++--)$ and $K: T \mathcal{M} \rightarrow T \mathcal{M}$ is an endomorphism such that
\begin{eqnarray}
\nonumber
(i) && K^{2} = id_{T \mathcal{M}}, \ \ \textrm{(K is paracomplex)}
\\ \nonumber
(ii) && \pm 1 \textrm{ eigenvalues of } K \textrm{ have rank 2}
\\ \nonumber
(iii) && ds^{2}(KX, KY) = - ds^{2}(X,Y) \ \textrm{for all } X,Y \in T \mathcal{M}, \ \ \textrm{(K is metric compatible)}
\\ \nonumber
(iv) && \textrm{the Nijenhuis tensor } N: T \mathcal{M} \otimes T \mathcal{M} \rightarrow T \mathcal{M},
\\ \nonumber
&& N(X,Y) := K[X, KY] + K[KX, Y] -[KX, KY] - [X,Y] 
\\ \nonumber
&& \textrm{ vanishes, i.e., } N(X,Y)=0 \textrm{ for all } X,Y \in T \mathcal{M}
\\ \nonumber
(v) && \textrm{a para-K\"ahler 2-form } \rho (X,Y) := ds^{2} (KX, Y) \textrm{ is closed, } d \rho =0.
\end{eqnarray}
\end{Definicja}

A few comments regarding Definition \ref{Definicja_para_Kahler} should be made. First, let us consider a pair of rank 2 distributions defined as follows
\begin{equation}
H := (K + id_{T \mathcal{M}}) \, T \mathcal{M}, \ \overline{H} := (K - id_{T \mathcal{M}}) \, T \mathcal{M}
\end{equation}
Hence, $T \mathcal{M} = H \oplus \overline{H}$. Both $H$ and $\overline{H}$ are null with respect to $ds^{2}$ and they are both SD or ASD. The condition $(iv)$ of the Definition \ref{Definicja_para_Kahler} assures, that both $H$ and $\overline{H}$ are integrable, i.e., $[H, H] \subset H$ and $[\overline{H}, \overline{H}] \subset \overline{H}$. Integral manifolds of $H$ and $\overline{H}$ are 2D, real, totally null surfaces which in what follows we call \textsl{null strings}. Also, by \textsl{a congruence (foliation) of null strings} we understand a family of null strings. Hence, para-K\"ahler spaces are equipped with two families of totally null surfaces (two congruences of null strings) of the same duality: one generated by the distribution $H$ and the other by the distribution $\overline{H}$. 

Moreover, the condition $(v)$ of the Definition \ref{Definicja_para_Kahler} is equivalent to the statement that both $H$ and $\overline{H}$ are parallely propagated. Spaces equipped with a null distribution which is parallely propagated are called \textsl{Walker spaces} \cite{Walker}. Hence, para-K\"ahler spaces are special examples of Walker spaces, and since they are equipped with a pair of such distributions they are sometimes called \textsl{double Walker spaces \cite{Chudecki_Ref_1}}. Following Pleba\'nski terminology, congruences of null strings which consist of integral manifolds of parallely propagated 2D distributions are called \textsl{nonexpanding}. Hence, we arrive at an alternative Definition of para-K\"ahler spaces
\begin{Definicja}
\label{Definicja_para_Kahler_alternative}
A para-K\"ahler space is a pair $(\mathcal{M}, ds^{2})$ where $\mathcal{M}$ is a 4-dimensional real smooth manifold, $ds^{2}$ is a real smooth metric of the neutral signature $(++--)$ and for each point $p \in \mathcal{M}$ there exists an open neighborhood $U$ of $p$ such that there exist two nonexpanding congruences of null strings of the same duality.
\end{Definicja}

It can be easily proved (see, e.g., \cite{Bor_Makhmali_Nurowski, Chudecki_notes_on_congr}) that if we choose an orientation in such a manner that both $H$ and $\overline{H}$ are SD, then Petrov-Penrose type of the SD Weyl tensor is [D], or - to be more precise (as in this Section we deal with 4D real neutral manifolds\footnote{Subtleties of Petrov-Penrose classification of the Weyl tensor related to differences between real neutral spaces and complex spaces are sketched in \cite{Bor_Makhmali_Nurowski,Chudecki_Ref_1}.}) - $[\textrm{D}_{r}]$. The ASD Weyl tensor is of an arbitrary Petrov-Penrose type. Thus, para-K\"ahler spaces are spaces of the types $[\textrm{D}_{r}]^{nn} \otimes [\textrm{any}]$, where double superscript $nn$ means that both congruences of SD null strings are nonexpanding. In what follows we are interested only in the case for which the ASD Weyl tensor is algebraically general, i.e., it is of the types $[\textrm{I}_{r}]$, $[\textrm{I}_{rc}]$ or $[\textrm{I}_{c}]$.  

At the end of this Section we remind the obvious: para-K\"ahler Einstein spaces are para-K\"ahler spaces which satisfy Einstein vacuum field equations with cosmological constant $\Lambda$
\begin{equation}
C_{ab} = 0, R = -4 \Lambda
\end{equation}
where $C_{ab}$ is the traceless Ricci tensor and $R$ is the curvature scalar.

\subsection{Hyperheavenly spaces of the types $[\textrm{D}]^{nn} \otimes [\textrm{any}]$}
The fact that pKE-space is a space of type $[\textrm{D}_{r}]^{nn} \otimes [\textrm{any}]$ satisfying Einstein vacuum field equations with $\Lambda$ is a bridge which leads to complex generalizations: \textsl{nonexpanding hyperheavenly spaces}. 

\subsubsection{Rudiments of hyperheavenly spaces}

\begin{Definicja}
\label{HH_space}
\textsl{A hyperheavenly space ($\mathcal{HH}$-space) with cosmological constant $\Lambda$} is a 4D complex analytic differential manifold endowed with a holomorphic metric $ds^{2}$ satisfying vacuum Einstein field equations with cosmological constant $\Lambda$ and such that the self-dual Weyl tensor is algebraically degenerate.
\end{Definicja}

The fact that an $\mathcal{HH}$-space is algebraically degenerate on SD side is equivalent to the existence of a congruence of SD null strings (by the virtue of \textsl{the complex Goldberg-Sachs theorem} \cite{Plebanski_Hacyan}). In what follows we are interested in type-$[\textrm{D}]^{nn} \otimes [\textrm{any}]$ $\mathcal{HH}$-spaces. The procedure how to specialize a type-$[\textrm{deg}]^{n} \otimes [\textrm{any}]$ $\mathcal{HH}$-space to a type-$[\textrm{D}]^{nn} \otimes [\textrm{any}]$ $\mathcal{HH}$-space is straightforward and we skip relevant details (they can be found in \cite{Finley_intrinsic,Chudecki_klasyfikacja_Killingow_nieeks}). 

The metric of an $\mathcal{HH}$-space of a type $[\textrm{D}]^{nn} \otimes [\textrm{any}]$ can be always brought to the form 
\begin{equation}
\label{metryka_HH_nieekspandujaca}
\frac{1}{2} ds^{2} = - dp^{\dot{A}}dq_{\dot{A}} + \left( - \Theta_{p_{\dot{A}}p_{\dot{B}}}  + \frac{\Lambda}{3} p^{\dot{A}} p^{\dot{B}} \right) dq_{\dot{A}} dq_{\dot{B}} = e^{1}e^{2} + e^{3}e^{4}, \ \ \ \dot{A}, \dot{B} = \dot{1}, \dot{2}
\end{equation}
where $(e^{1},e^{2},e^{3},e^{4})$ is a null tetrad called \textsl{the Pleba\'nski tetrad}
\begin{equation}
\label{tetrada_Plebanskiego__HH_nieekspandujaca}
[e^{3}, e^{1}] = dq_{\dot{A}} , \ [e^{4}, e^{2}] = -dp^{\dot{A}} + \left( - \Theta_{p_{\dot{A}}p_{\dot{B}}}  + \frac{\Lambda}{3} p^{\dot{A}} p^{\dot{B}} \right) dq_{\dot{B}}
\end{equation}
$(q^{\dot{A}}, p^{\dot{B}})$ are local coordinates, $\Theta = \Theta (q^{\dot{A}}, p^{\dot{B}})$ is a holomorphic function called \textsl{the key function} and $\Lambda$ is a cosmological constant. In Pleba\'nski tetrad SD conformal curvature coefficients read
\begin{equation}
\label{krzywizna_nieekspandujaca}
C^{(1)} = C^{(2)} = C^{(4)} = C^{(5)}=0,  \ C^{(3)} = - \frac{2}{3} \Lambda 
\end{equation}
In what follows we assume that $\Lambda \ne 0$ holds true, otherwise $C^{(3)}=0$ and a space becomes a heavenly space. ASD curvature takes the nice form
\begin{equation}
\label{krzywizna_nieekspandujaca_ASD}
 C_{\dot{A}\dot{B}\dot{C}\dot{D}} = \Theta_{p^{\dot{A}}p^{\dot{B}}p^{\dot{C}}p^{\dot{D}}}
\end{equation}
Vacuum Einstein field equations with $\Lambda$ are reduced to the nonlinear, second order PDE called \textsl{the (nonexpanding) hyperheavenly equation}
\begin{equation}
\label{rownanie_HH_nieekspadujace}
\frac{1}{2} \Theta_{p_{\dot{A}}p_{\dot{B}}}\Theta_{p^{\dot{A}}p^{\dot{B}}} + \Theta_{p_{\dot{A}} q^{\dot{A}}} + \Lambda \left( p^{\dot{A}} \Theta_{p^{\dot{A}}} - \Theta - \frac{1}{3} p^{\dot{A}} p^{\dot{B}} \Theta_{p^{\dot{A}} p^{\dot{B}}} \right) =0
\end{equation}

It is sometimes useful to abandon the coordinates $(q^{\dot{A}}, p^{\dot{B}})$ written in a spinor-like convention and define the following abbreviations
\begin{equation}
q:= q^{\dot{1}}, \ p:= q^{\dot{2}}, \ x:=p^{\dot{1}}, \ y:=p^{\dot{2}}
\end{equation}
We refer to the coordinates $(q,p,x,y)$ as \textsl{the hyperheavenly coordinates}. The metric (\ref{metryka_HH_nieekspandujaca}) written in the hyperheavenly coordinates takes the form
\begin{eqnarray}
\label{metryka_w_xyqp}
\frac{1}{2} ds^{2} &=& dydq-dxdp  - \left( \Theta_{xx} - \frac{\Lambda}{3}y^{2} \right) dq^{2} - 2\left( \Theta_{xy} + \frac{\Lambda}{3}xy \right) dqdp
\\ \nonumber
&& \ \ \ \ \ \ \ \ \ \ \ \ \ \ \ \ \ \ 
- \left( \Theta_{yy} - \frac{\Lambda}{3}x^{2} \right) dp^{2} 
\end{eqnarray}
The hyperheavenly equation reads
\begin{eqnarray}
\label{rownanie_hiperniebianskie}
&& \Theta_{xx} \Theta_{yy} - \Theta_{xy}^{2}+ \Theta_{qy} - \Theta_{px}  
\\ 
\nonumber
&&   \ \ \ \ \ \ \ \ 
+ \Lambda \bigg( x \Theta_{x} + y \Theta_{y} - \Theta - \frac{1}{3} x^{2} \Theta_{xx} - \frac{1}{3} y^{2} \Theta_{yy} - \frac{2}{3} xy \Theta_{xy}   \bigg) = 0
\end{eqnarray}
Finally, the ASD curvature coefficients are easily expressible as the fourth partial derivatives of the function $\Theta$
\begin{eqnarray}
\label{wspolczynniki_ASD}
&&  \dot{C}^{(5)} := 2 C_{\dot{1} \dot{1} \dot{1} \dot{1}} = 2\Theta_{xxxx}
\\ \nonumber
&&  \dot{C}^{(4)} := 2 C_{\dot{1} \dot{1} \dot{1} \dot{2}} = 2\Theta_{xxxy}
\\ \nonumber
&&  \dot{C}^{(3)} := 2 C_{\dot{1} \dot{1} \dot{2} \dot{2}} = 2\Theta_{xxyy}
\\ \nonumber
&&  \dot{C}^{(2)} := 2 C_{\dot{1} \dot{2} \dot{2} \dot{2}} = 2\Theta_{xyyy}
\\ \nonumber
&&  \dot{C}^{(1)} := 2 C_{\dot{2} \dot{2} \dot{2} \dot{2}} = 2\Theta_{yyyy}
\end{eqnarray}

\subsubsection{Gauge freedom}

The metric (\ref{metryka_HH_nieekspandujaca}) remains invariant under the following transformations of the variables 
\begin{equation}
\label{gauge_wspolrzednych}
 q'_{\dot{A}} = q'_{\dot{A}} (q_{\dot{B}}), \ 
p'^{\dot{A}} = \Delta^{-1} \frac{\partial q'^{\dot{A}}}{\partial q^{\dot{B}}} \left( p^{\dot{B}}  -\frac{1}{\Lambda} \frac{\partial \ln \Delta}{\partial q_{\dot{B}}} \right) , \ \Delta := \det \left(  \frac{\partial q'_{\dot{A}}}{\partial q_{\dot{B}}} \right)
\end{equation}
Under (\ref{gauge_wspolrzednych}) the key function transforms as follows
\begin{equation}
\label{transformacja_funkcji_kluczowej}
 \Delta^{2} \Theta' = \Theta - \frac{1}{6} L_{\dot{A}\dot{B}\dot{C}} \, p^{\dot{A}} p^{\dot{B}} p^{\dot{C}} + \frac{1}{2}N_{\dot{A}\dot{B}} \, p^{\dot{A}} p^{\dot{B}} -H_{\dot{A}} \, p^{\dot{A}} + \frac{1}{\Lambda} \left( \frac{\partial H^{\dot{A}}}{\partial q^{\dot{A}}} - \frac{1}{2} N_{\dot{A}\dot{B}} N^{\dot{A}\dot{B}} \right)
 \end{equation}
where $H_{\dot{A}} = H_{\dot{A}} (q^{\dot{B}})$ are arbitrary functions and
\begin{eqnarray}
&& L_{\dot{A}\dot{B}\dot{C}} = L_{(\dot{A}\dot{B}\dot{C})} := \Delta^{-1} \frac{\partial q'_{\dot{X}}}{\partial q^{(\dot{A}}} \frac{\partial^{2} q'^{\dot{X}}}{\partial q^{\dot{B}} \partial q^{\dot{C})}}
\\ \nonumber
&& N_{\dot{A}\dot{B}} = N_{(\dot{A}\dot{B})} := \frac{1}{3 \Lambda} \frac{\partial \ln \Delta}{\partial q^{\dot{A}}}\frac{\partial \ln \Delta}{\partial q^{\dot{B}}} + \frac{1}{\Lambda} \frac{\partial^{2} \ln \Delta}{\partial q^{\dot{A}} \partial q^{\dot{B}}} - \frac{1}{\Lambda}  \frac{\partial \ln \Delta}{\partial q'_{\dot{X}}} \frac{\partial^{2} q'_{\dot{X}}}{\partial q^{\dot{A}} \partial q^{\dot{B}}} 
\end{eqnarray}
Hence, we have four arbitrary gauge functions of two variables, $q'^{\dot{A}} = q'^{\dot{A}} (q^{\dot{B}})$ and $H_{\dot{A}} = H_{\dot{A}} (q^{\dot{B}})$ at our disposal.

\subsubsection{Symmetries}

Symmetries in $\mathcal{HH}$-spaces of the types $[\textrm{D}]^{nn} \otimes [\textrm{any}]$ generated by vectors $K_{a}$ such that the set of equations $\nabla_{(a} K_{b)}=\chi \, g_{ab}$ is satisfied were considered in \cite{Plebanski_Finley_Killing} and then revisited in \cite{Chudecki_klasyfikacja_Killingow_nieeks}. It was proved that any Killing vector admitted by a space of the types $[\textrm{D}]^{nn} \otimes [\textrm{any}]$ has the form
\begin{equation}
\label{ogolna_postac_wektora_Killinga}
K = \delta^{\dot{B}} \frac{\partial}{\partial q^{\dot{B}}} + \left( \frac{\partial \delta^{\dot{M}}}{\partial q_{\dot{B}}} p_{\dot{M}} + \epsilon^{\dot{B}} \right) \frac{\partial}{\partial p^{\dot{B}}} , \ \epsilon^{\dot{A}} := -\frac{1}{\Lambda} \frac{\partial}{\partial q_{\dot{A}}} \left( \frac{\partial \delta^{\dot{B}}}{\partial q^{\dot{B}}}  \right)
\end{equation}
where $\delta^{\dot{A}} = \delta^{\dot{A}} (q^{\dot{B}})$. Killing equations are reduced to a single \textsl{master equation}
\begin{equation}
\label{master_equation}
K (\Theta) = -2  \frac{\partial \delta^{\dot{N}}}{\partial q^{\dot{N}}} \, \Theta   +  \frac{1}{6} \frac{\partial^{2} \delta_{\dot{A}}}{\partial q^{\dot{B}} \partial q^{\dot{C}}} \, p^{\dot{A}}p^{\dot{B}}p^{\dot{C}} 
+ \frac{1}{2} \frac{\partial \epsilon_{\dot{A}}}{\partial q^{\dot{B}}} \, p^{\dot{A}}p^{\dot{B}} + \zeta_{\dot{A}} \, p^{\dot{A}} -\frac{1}{\Lambda} \frac{\partial \zeta^{\dot{A}}}{\partial q^{\dot{A}}}
\end{equation}
where $\zeta^{\dot{A}}$ are arbitrary functions of $q^{\dot{B}}$. 
\begin{Uwaga}
\label{Uwaga_Killingi_1}
\normalfont
Let us remind that a vector $K_{a}$ which satisfies a set of equations $\nabla_{(a} K_{b)}=\chi \, g_{ab}$ is called \textsl{a proper conformal vector} if $\chi$ is not a constant, \textsl{a proper homothetic vector} if $\chi = \textrm{const} \ne 0$ and \textsl{a Killing vector} if $\chi=0$. It was proved in \cite{Chudecki_klasyfikacja_Killingow_nieeks} that in $\mathcal{HH}$-spaces of the types $[\textrm{D}]^{nn} \otimes [\textrm{any}]$ the factor $\chi$ must vanish, $\chi = 0$. Thus, only Killing vectors are admitted.
\end{Uwaga}
\begin{Uwaga}
\label{Uwaga_Killingi_2}
\normalfont
It can be checked by a simple calculation that any null Killing vector with $\delta^{\dot{A}}=0$ must vanish, while the postulate of a null Killing vector with $\delta^{\dot{A}} \ne 0$ implies $\Lambda = 0$. Thus, null Killing vectors are not admitted in our setting.  
\end{Uwaga}
From Remarks \ref{Uwaga_Killingi_1} and \ref{Uwaga_Killingi_2} it follows that any algebra of infinitesimal symmetries admitted by type-$[\textrm{D}]^{nn} \otimes [\textrm{any}]$ $\mathcal{HH}$-spaces consists of nonnull Killing vectors only.

Under (\ref{gauge_wspolrzednych}) the functions $\delta^{\dot{A}}$ and $\zeta^{\dot{A}}$ transform as follows
\begin{subequations}
\begin{eqnarray} 
\label{transformacja_delty}
\delta'^{\dot{M}} &=& \frac{\partial q'^{\dot{M}}}{\partial q^{\dot{B}}} \delta^{\dot{B}}
\\
\label{transformacja_zeta}
\Delta^{2} \zeta'_{\dot{R}} &=&  \frac{\partial q'_{\dot{R}}}{\partial q_{\dot{N}}} \left( \zeta_{\dot{N}} - \delta^{\dot{M}} \frac{\partial H_{\dot{N}}}{\partial q^{\dot{M}}} + H_{\dot{M}} \frac{\partial \delta_{\dot{N}}}{\partial q_{\dot{M}}} +\epsilon^{\dot{M}} N_{\dot{M}\dot{N}} -2 H_{\dot{N}} \frac{\partial \delta^{\dot{M}}}{\partial q^{\dot{M}}}  \right) \ \ \ \ \ \ 
\\ \nonumber
&& -\frac{\Delta^{2}}{\Lambda^{2}} \frac{\partial \ln \Delta}{\partial q'_{\dot{S}}} \frac{\partial \ln \Delta}{\partial q'_{\dot{T}}} \frac{\partial^{2} \delta'_{(\dot{R}}}{\partial q'^{\dot{S}} \partial q'^{\dot{T})}}  +  \frac{\Delta^{2}}{\Lambda} \frac{\partial \ln \Delta}{\partial q'_{\dot{M}}}   \frac{\partial \epsilon'_{(\dot{M}}}{\partial q'^{\dot{R})}} 
\end{eqnarray}
\end{subequations}

Let us write a Killing vector in the hyperheavenly coordinates
\begin{eqnarray}
K &=& \delta^{\dot{1}} \, \frac{\partial}{\partial q} +  \delta^{\dot{2}} \, \frac{\partial}{\partial p} 
+ \left( \frac{\partial \delta^{\dot{1}}}{\partial p} y - \frac{\partial \delta^{\dot{2}}}{\partial p} x  + \epsilon^{\dot{1}}\right)  \frac{\partial}{\partial x}
+ \left( -\frac{\partial \delta^{\dot{1}}}{\partial q} y + \frac{\partial \delta^{\dot{2}}}{\partial q} x + \epsilon^{\dot{2}} \right)  \frac{\partial}{\partial y} \ \ \ \ \ \ 
\\ \nonumber
&& \epsilon^{\dot{1}} = - \frac{1}{\Lambda} \frac{\partial \delta}{\partial p}, \, \epsilon^{\dot{2}} =  \frac{1}{\Lambda} \frac{\partial \delta}{\partial q}, \, \delta := \frac{\partial \delta^{\dot{M}}}{\partial q^{\dot{M}}}
\end{eqnarray}
The master equation reads then
\begin{eqnarray}
K(\Theta) &=& -2 \delta \Theta + \frac{1}{6} \frac{\partial^{2} \delta^{\dot{2}}}{\partial q^{2}}x^{3} -
\frac{1}{6} \frac{\partial^{2} \delta^{\dot{1}}}{\partial p^{2}}y^{3} +
\frac{1}{6} \left( 2 \frac{\partial^{2} \delta^{\dot{2}}}{\partial q \partial p} - \frac{\partial^{2} \delta^{\dot{1}}}{\partial q^{2}}  \right)  x^{2}y
\\ \nonumber
&& +\frac{1}{6} \left( \frac{\partial^{2} \delta^{\dot{2}}}{\partial p^{2}} - 2 \frac{\partial^{2} \delta^{\dot{1}}}{\partial q \partial p}   \right)  xy^{2} + \frac{1}{2\Lambda} \frac{\partial^{2} \delta}{\partial q^{2}}x^{2}+ \frac{1}{\Lambda} \frac{\partial^{2} \delta}{\partial q \partial p} xy + \frac{1}{2\Lambda} \frac{\partial^{2} \delta}{\partial p^{2}}y^{2}
\\ \nonumber
&& + \zeta_{\dot{1}} x + \zeta_{\dot{2}} y - \frac{1}{\Lambda} \left(  \frac{\partial \zeta^{\dot{1}}}{\partial q} + \frac{\partial \zeta^{\dot{2}}}{\partial p}  \right)
\end{eqnarray}

\subsection{Criteria for the ASD Weyl tensor to be algebraically general}

Let $\xi^{\dot{A}}$ be an arbitrary spinor such that $\xi^{\dot{2}} \ne 0$. If we define $\xi: = \xi^{\dot{1}} / \xi^{\dot{2}}$ then $C_{\dot{A}\dot{B}\dot{C}\dot{D}} \xi^{\dot{A}} \xi^{\dot{B}} \xi^{\dot{C}} \xi^{\dot{D}}$ becomes a fourth-order polynomial in $\xi$. Indeed
\begin{eqnarray}
\mathcal{C}(\xi) &:=& C_{\dot{A}\dot{B}\dot{C}\dot{D}} \xi^{\dot{A}} \xi^{\dot{B}} \xi^{\dot{C}} \xi^{\dot{D}}
\\ \nonumber
 &=& (\xi^{\dot{2}})^{4} ( \Theta_{xxxx} \, \xi^{4} + 4 \Theta_{xxxy} \, \xi^{3} + 6 \Theta_{xxyy} \, \xi^{2} + 4 \Theta_{xyyy} \, \xi + \Theta_{yyyy} )
 \\ \nonumber
  &=& \frac{1}{2}(\xi^{\dot{2}})^{4} ( \dot{C}^{(5)} \xi^{4} + 4 \dot{C}^{(4)} \xi^{3} + 6 \dot{C}^{(3)} \xi^{2} + 4 \dot{C}^{(2)} \xi + \dot{C}^{(1)} )
\end{eqnarray}
The Petrov-Penrose classification is related to the multiplicity of roots of the polynomial $\mathcal{C}(\xi)$. To be more precise, the ASD conformal curvature is algebraically general if and only if $\mathcal{C}(\xi)$ has four different roots. \textsl{A discriminant} $\mathcal{D}$ is defined as 
\begin{equation}
\mathcal{D} := I^{3} - 27 J^{2}
\end{equation}
where 
\begin{eqnarray}
I &:=& \dot{C}^{(1)} \dot{C}^{(5)} - 4 \dot{C}^{(2)} \dot{C}^{(4)} + 3 \dot{C}^{(3)} \dot{C}^{(3)}
\\ \nonumber
J &:=& \left| \begin{array}{ccc}
\dot{C}^{(5)} & \dot{C}^{(4)} & \dot{C}^{(3)} \\
\dot{C}^{(4)} & \dot{C}^{(3)} & \dot{C}^{(2)} \\
\dot{C}^{(3)} & \dot{C}^{(2)} & \dot{C}^{(1)} 
\end{array}\right| 
\\ \nonumber
&=&\dot{C}^{(5)} (\dot{C}^{(3)}\dot{C}^{(1)} - \dot{C}^{(2)}\dot{C}^{(2)}) - \dot{C}^{(4)} (\dot{C}^{(4)}\dot{C}^{(1)}-\dot{C}^{(2)}\dot{C}^{(3)}) + \dot{C}^{(3)} (\dot{C}^{(4)} \dot{C}^{(2)}-\dot{C}^{(3)}\dot{C}^{(3)})
\end{eqnarray}
In a complex case the ASD Weyl tensor is type-[I] if and only if $\mathcal{D} \ne 0$. In a real neutral case $\mathcal{D}$ also must be nonzero but there are three different algebraically general types, namely $[\textrm{I}_{r}]$ (all roots are real), $[\textrm{I}_{c}]$ (all roots are complex) and $[\textrm{I}_{rc}]$ (two real roots and a pair of mutually conjugated complex roots). To distinguish these types we define quantities 
\begin{eqnarray}
\label{definicje_P_i_R}
\mathcal{P} &:=& 48 ( \dot{C}^{(3)}\dot{C}^{(5)} - \dot{C}^{(4)} \dot{C}^{(4)})
\\ \nonumber
\mathcal{R} &:=& 64 \big[ \dot{C}^{(1)} (\dot{C}^{(5)})^{3} - 9(\dot{C}^{(3)})^{2} (\dot{C}^{(5)})^{2} + 24 \dot{C}^{(3)}\dot{C}^{(5)} (\dot{C}^{(4)})^{2} 
\\ \nonumber
&& \ \ \ \ \ \ \ -4 \dot{C}^{(2)}\dot{C}^{(4)} (\dot{C}^{(5)})^{2} - 12 (\dot{C}^{(4)})^{4}  \big]
\end{eqnarray}
which enter into the criteria gathered in Table \ref{kryteria_niezdegenerowania}.
\begin{table}[ht]
\begin{center}
\begin{tabular}{|c|c|c|c|}   \hline
\multicolumn{2}{|c|}{Complex case }  & \multicolumn{2}{|c|}{Real neutral case } \\ \hline
Type & Criterion & Type &  Criteria  \\  \hline
 &  & $[\textrm{I}_{rc}]$ &  $\mathcal{D} < 0$ \\  \cline{3-4}
$[\textrm{I}] $ & $\mathcal{D} \ne 0$ & $[\textrm{I}_{r}]$ & $\mathcal{D} > 0$; $\mathcal{P}<0$ and $\mathcal{R}<0$  \\  \cline{3-4}
 & & $[\textrm{I}_{c}]$ & $\mathcal{D} > 0$; $\mathcal{P}>0$ or $\mathcal{R}>0$ \\ \hline
\end{tabular}
\caption{Criteria for determining type of an algebraically general ASD Weyl tensor.}
\label{kryteria_niezdegenerowania}
\end{center}
\end{table}

\subsection{An adaptation of a null tetrad}

It is well-known that there exists a canonical adaptation of a null tetrad in the case when ASD Weyl tensor is algebraically general (see \cite{Plebanski_spinors}). Indeed, the ASD Weyl spinor can be decomposed into a product of 1-index dotted spinors which are called \textsl{dotted Penrose spinors}
\begin{equation}
C_{\dot{A}\dot{B}\dot{C}\dot{D}} = \alpha_{(\dot{A}} \beta_{\dot{B}} \gamma_{\dot{C}} \delta_{\dot{D})}
\end{equation}
Penrose spinors $\alpha_{\dot{A}}$, $\beta_{\dot{A}}$, $\gamma_{\dot{A}}$ and $\delta_{\dot{A}}$ are pairwise linearly independent. Let a basis of dotted spinors be given by a pair of spinors $(a_{\dot{A}}, b_{\dot{B}})$, $a_{\dot{A}} b^{\dot{B}} \ne 0$. Spinors $a_{\dot{A}}$ and $b_{\dot{A}}$ can be chosen arbitrarily, so one can choose $a_{\dot{A}}$ to be proportional to $\alpha_{\dot{A}}$ and $b_{\dot{A}}$ to be proportional to $\beta_{\dot{A}}$. This choice implies $\dot{C}^{(1)}=0=\dot{C}^{(5)}$. 

However, the canonical choice $\dot{C}^{(1)}=0=\dot{C}^{(5)}$ is possible only if a null tetrad is not fixed and null rotations are not restricted. Unfortunately, Pleba\'nski tetrad is partially fixed. Indeed, gauge (\ref{gauge_wspolrzednych}) induces the following admissible transformations of dotted 1-index spinors
\begin{equation}
\label{dotowane_transformacje_spinorowe}
\alpha'^{\dot{A}} = \Delta^{-\frac{1}{2}} \frac{\partial q'^{\dot{A}}}{\partial q^{\dot{B}}} \, \alpha^{\dot{B}}
\end{equation}
The choice $2\Theta_{yyyy} = \dot{C}^{(1)}=0$ requires $\alpha_{\dot{2}} \equiv -\alpha^{\dot{1}} =0$, but from (\ref{dotowane_transformacje_spinorowe}) we find
\begin{equation}
\alpha'^{\dot{1}}  = \Delta^{-\frac{1}{2}} \left( \frac{\partial q'}{\partial q} \alpha^{1} + \frac{\partial q'}{\partial p} \alpha^{2} \right) 
\end{equation}
Now $\alpha^{\dot{1}}$ is a function of 4 variables, $\alpha^{\dot{1}}=\alpha^{\dot{1}} (q,p,x,y)$, while $q'$ depends on $q$ and $p$ only.  Thus, the condition $\alpha^{\dot{1}}=0$ cannot be, in general, obtained by means of the transformation rule (\ref{dotowane_transformacje_spinorowe}).  Nevertheless, it can be considered as an additional \textsl{ad hoc} assumption which (possibly) could lead to an example of algebraically general solution. Detailed analysis proves that it is not a case, since we have
\begin{Lemat}
\label{lemat_o_nieznikaniu_C1}
Any solution of the equation (\ref{rownanie_hiperniebianskie}) such that $\Theta_{yyyy}=0$ leads to  $\mathcal{D}=0$.
\end{Lemat}
\begin{proof}[Sketch of the Proof.] The condition $\Theta_{yyyy}=0$ implies that the key function $\Theta$ becomes a third order polynomial in $y$ with coefficients depending on $(q,p,x)$. Inserting this form of $\Theta$ into the $\mathcal{HH}$-equation produces a system of four equations. Then the system splits into several subcases. All of these subcases can be solved but corresponding $\Theta$ always gives $\mathcal{D}=0$. 
\end{proof}
The Lemma \ref{lemat_o_nieznikaniu_C1} yields
\begin{Wniosek}
\label{wniosek_o_niezerowosci}
Any solution of the equation (\ref{rownanie_hiperniebianskie}) such that $\Theta_{xxxx}=0$ leads to $\mathcal{D}=0$.
\end{Wniosek}
\begin{proof}
The $\mathcal{HH}$-metric (\ref{metryka_w_xyqp}) and the $\mathcal{HH}$-equation (\ref{rownanie_hiperniebianskie}) are symmetric in the coordinates $x$ and $y$. It is enough to interchange $x$ and $y$ and the proof becomes identical to that of the Lemma \ref{lemat_o_nieznikaniu_C1}.
\end{proof}


\section{Algebraically general para-K\"ahler Einstein spaces with symmetries}
\label{Section_3}
\setcounter{equation}{0}

Obtaining the general analytic solution of equation (\ref{rownanie_hiperniebianskie}) seems to be beyond bounds of possibility and, as we have seen, the most promising simplifying assumption leads to algebraically degenerate pKE-spaces. Thus, we need more systematic approach, and in this Section we consider gpKE-spaces equipped with additional Killing symmetries.

\subsection{Space with a Killing vector}
\label{single_Killing_vector}

Let us assume that gpKE-space admits a Killing vector. Then we have
\begin{Lemat}
\label{lemat_o_pierwszym_Killingu}
If a space of the type $[\textrm{D}]^{nn} \otimes [\textrm{I}]$ admits a Killing vector $K_{1}$ then $K_{1}$ can be always brought to the form
\begin{equation}
\label{pierwszy_Killing}
K_{1} = \frac{\partial}{\partial q}
\end{equation}
and the key function $\Theta$ reduces to a function of only three variables, $\Theta = \Theta (p,x,y)$.
\end{Lemat}
\begin{proof}
Any Killing vector admitted by a gpKE-space must be of the form (\ref{ogolna_postac_wektora_Killinga}) with $\delta^{\dot{1}} \ne 0$ or $\delta^{\dot{2}} \ne 0$. Assume, that $\delta^{\dot{1}} \ne 0$ holds true. Then one puts $\delta^{\dot{1}} =1 $ and $\delta^{\dot{2}} = 0$ without any loss of generality (compare transformation formula (\ref{transformacja_delty})). This proves (\ref{pierwszy_Killing}). From (\ref{transformacja_zeta}) it follows that functions $\zeta^{\dot{A}}$ can be gauged away. The master equation yields
\begin{equation}
\frac{\partial \Theta}{\partial q} = 0  \ \Longleftrightarrow \ \Theta = \Theta (p,x,y)
\end{equation}
\end{proof}  
Unfortunately, with only one symmetry assumed, the $\mathcal{HH}$-equation (\ref{rownanie_hiperniebianskie}) is slightly simplified. We conclude that a single symmetry is (probably) not enough for a reasonable simplification of the problem.

\begin{Uwaga}
\normalfont
The choice $K_{1} = \partial_{q}$ restricts the gauge (\ref{gauge_wspolrzednych})  to the conditions
\begin{eqnarray}
\label{gaug_pozostaly_po_pierwszym_Killingu}
&& q'=q + \mathfrak{f}(p), \ p'=p'(p), \ \frac{dp'}{dp} =: \mathfrak{h}(p), \ \Delta = \mathfrak{h} (p)
\\ \nonumber
&& H_{\dot{N}} = H_{\dot{N}} (p), \ N_{\dot{1}\dot{1}}=N_{\dot{1}\dot{2}}=0, \ N_{\dot{2}\dot{2}} = \frac{\mathfrak{h}^{\frac{2}{3}}}{\Lambda} \frac{\partial}{\partial p} (\mathfrak{h}^{-\frac{5}{3}} \mathfrak{h}_{p})
\\ \nonumber
&& L_{\dot{1}\dot{1}\dot{1}} = L_{\dot{1}\dot{1}\dot{2}} = 0, \ L_{\dot{1}\dot{2}\dot{2}} = -\frac{1}{3} \frac{\mathfrak{h}_{p}}{\mathfrak{h}}, \ L_{\dot{2}\dot{2}\dot{2}} = \frac{\mathfrak{h} \mathfrak{f}_{pp}-\mathfrak{h}_{p} \mathfrak{f}_{p}}{\mathfrak{h}}
\end{eqnarray}
\end{Uwaga}

\subsection{Space with two commuting Killing vectors}
\label{subsection_two_kommuting_Killing_vectors}

Let us assume now that a gpKE-space admits 2D algebra of commuting Killing vector fields. For this case one formulates
\begin{Lemat}
\label{lemat_o_drugim_Killingu}
If a space of the type $[\textrm{D}]^{nn} \otimes [\textrm{I}]$ admits 2D trivial algebra $\mathcal{A}$ of commuting Killing vector fields, then for arbitrary basis $\{K_{1}, K_{2} \}$ of $\mathcal{A}$ one can choose coordinate system $(q,p,x,y)$ to bring this basis to the form 
\begin{equation}
K_{1} = \frac{\partial}{\partial q} \quad \quad K_{2} = \frac{\partial}{\partial p}
\end{equation}
Moreover, in this coordinate system the key function $\Theta$ reduces to a function of only two variables, $\Theta = \Theta (x,y)$.
\end{Lemat}

\begin{proof}
The first Killing vector can be always brought to the form $K_{1} = \partial_{q}$ (see Lemma \ref{lemat_o_pierwszym_Killingu}). The second Killing vector $K_{2}$ has the form (\ref{ogolna_postac_wektora_Killinga}) and from the condition $[K_{1}, K_{2}]=0$ it follows that $\delta^{\dot{B}} = \delta^{\dot{B}} (p)$. From (\ref{transformacja_delty}) and (\ref{gaug_pozostaly_po_pierwszym_Killingu}) we find the transformation formulas for $\delta^{\dot{B}}$ of $K_{2}$
\begin{equation}
\label{zredukowane_wzory_trans_na_delta}
\delta'^{\dot{1}} = \delta^{\dot{1}} + \frac{d\mathfrak{f}}{dp} \, \delta^{\dot{2}}, \ \delta'^{\dot{2}} = \mathfrak{h} \, \delta^{\dot{2}}
\end{equation}

Consider first a case with $\delta^{\dot{2}} = 0$. Then $K_{2} = \delta^{\dot{1}} \, \partial_{q} + \dfrac{\partial \delta^{\dot{1}}}{\partial p} y \, \partial_{x}$ where $\delta^{\dot{1}} = \delta^{\dot{1}} (p)$. Inserting this form of $K_{2}$ into the master equation and differentiating it with respect to $x$ three times we find that $\dfrac{\partial \delta^{\dot{1}}}{\partial p}  \, \Theta_{xxx}=0$. Thus, $\Theta_{xxx}=0$ or $\delta^{\dot{1}} =\delta^{\dot{1}}_{0}= \textrm{const}$. However, both these possibilities lead to a contradiction. Indeed, $\Theta_{xxx}=0$ implies $C^{(5)}=C^{(4)}=0$ and the space is not a gpKE-space anymore (see Corollary \ref{wniosek_o_niezerowosci}). The second possibility implies $K_{2} =\delta^{\dot{1}}_{0} \, \partial_{q} = \delta^{\dot{1}}_{0} K_{1}$. Thus, $K_{1}$ and $K_{2}$ are not independent. Hence, $\delta^{\dot{2}}$ is necessarily nonzero.

With $\delta^{\dot{2}} \ne 0$ assumed, it follows from (\ref{zredukowane_wzory_trans_na_delta}) that one can always put $\delta^{\dot{2}}=1$, $\delta^{\dot{1}}=0$. This yields $K_{2} = \partial_{p}$. 

The master equation reduces to the form
\begin{equation}
\label{zredukowane_rownanie_master}
\frac{\partial \Theta}{\partial p} = \zeta_{\dot{A}} p^{\dot{A}} - \frac{1}{\Lambda} \frac{\partial \zeta^{\dot{A}}}{\partial q^{\dot{A}}}
\end{equation}
Since $\Theta_{q}=0$, from (\ref{zredukowane_rownanie_master}) we find that $\zeta_{\dot{A}} = \zeta_{\dot{A}} (p)$. It means that using the gauge functions $H_{\dot{N}} = H_{\dot{N}} (p)$ one can always gauge away $\zeta_{\dot{A}}$. Hence, from (\ref{zredukowane_rownanie_master}) one gets $\Theta_{p}=0 \ \Longleftrightarrow \ \Theta = \Theta (x,y)$.
\end{proof}
 
\begin{Uwaga}
\normalfont
When $K_1$ and $K_2$ are fixed and brought to  $K_{1} = \partial_{q}$ and $K_{2} = \partial_{p}$, then the gauge (\ref{gauge_wspolrzednych}) becomes restricted to the transformations
\begin{equation}
\label{gaug_pozostaly_po_drugim_Killingu}
 q'=q + q_{0}, \ p'=p+p_{0}, \  \Delta = 1, \ H_{\dot{N}} = \textrm{const}, \ N_{\dot{A}\dot{B}}= L_{\dot{A}\dot{B}\dot{C}} = 0
\end{equation}
Hence, the remaining gauge is very limited: there are only four constants available. 
\end{Uwaga}

With $K_{1} = \partial_{q}$ and $K_{2} = \partial_{p}$, the $\mathcal{HH}$-equation reduces to the form
\begin{equation}
\label{rownanie_hiperniebianskie_dwa_komutujace_Killingi}
 \Theta_{xx} \Theta_{yy} - \Theta_{xy}^{2}
+ \Lambda \bigg( x \Theta_{x} + y \Theta_{y} - \Theta - \frac{1}{3} x^{2} \Theta_{xx} - \frac{1}{3} y^{2} \Theta_{yy} - \frac{2}{3} xy \Theta_{xy}   \bigg) = 0
\end{equation}
Eq. (\ref{rownanie_hiperniebianskie_dwa_komutujace_Killingi}) is still hard to solve. The most obvious step is an assumption that the key function separates, $\Theta (x,y) = F(x) G(y)$. However, this implies $\mathcal{D}=0$. (We are not going to give a proof of this fact here, as it is quite long, purely algebraical, but otherwise straightforward). 

One may suspect that a symmetry algebra $\mathcal{A}=2A_{1}$ could be too strong and could cause an algebraic degeneracy of the ASD Weyl tensor. However, there is no obvious reason for all solutions of (\ref{rownanie_hiperniebianskie_dwa_komutujace_Killingi}) to give $\mathcal{D}=0$. Consequently, in the next Section we consider  cases with even one more symmetry.

\subsection{Spaces with a 3D algebra of infinitesimal symmetries}
\label{three_Killings}

Let us recall the classification of real 3D Lie algebras as given in Table \ref{Lie_Algebras_table}. (We follow Ref. \cite{Patera} and the Table I therein). Since we want to employ results of the previous Section our considerations are restricted to 3D algebras of Killing vectors which admit 2D trivial subalgebra $\mathcal{A}$. Thus, algebras $A_{3,8}$ and $A_{3,9}$ are excluded from the very beginning. We are going to consider an arbitrary Killing vector $K_3$ not belonging to the algebra $\mathcal{A}$, and to identify remaining cases that can be relevant for algebraically nondegenerate gpKE-spaces.

\begin{table}[ht]
\begin{center}
\begin{tabular}{|c|c|c|}   \hline
 Name  &   Nonzero commutation relations & Comments          \\ \hline \hline
 $3A_{1}$ &  & Abelian   \\ \hline
 $A_{2,1} \oplus A_{1}$ & $[e_1, e_3] = e_1$  & decomposable, non-Abelian   \\ \hline
 $A_{3,1}$ & $[e_2, e_3] = e_1$  &  nilpotent  \\ \hline
 $A_{3,2}$ & $[e_1, e_3]=e_1$, $[e_2, e_3]=e_1+e_2$ &   solvable \\ \hline
 $A_{3,3}$ & $[e_1, e_3]=e_1$, $[e_2, e_3]=e_2$ &   solvable \\ \hline
 $A_{3,4}$ & $[e_1, e_3]=e_1$, $[e_2, e_3]=-e_2$ &   solvable \\ \hline
 $A_{3,5}^{a}$ & $[e_1, e_3]=e_1$, $[e_2, e_3]=ae_2$, $0< |a| <1$ &   solvable \\ \hline
 $A_{3,6}$ & $[e_1, e_3]=-e_2$, $[e_2, e_3]=e_1$ &   solvable \\ \hline
 $A_{3,7}^{a}$ & $[e_1, e_3]=ae_1-e_2$, $[e_2, e_3]=e_1+ae_2$, $a>0$ &   solvable \\ \hline
 $A_{3,8}$ & $[e_1, e_3]=-2e_2$, $[e_1, e_2]=e_1$, $[e_2, e_3]=e_3$ &   semisimple \\ \hline
 $A_{3,9}$ & $[e_1, e_2]=e_3$, $[e_2, e_3]=e_1$, $[e_3, e_1]=e_2$ &   semisimple \\ \hline
\end{tabular}
\caption{Real Lie algebras of dimensions three.}
\label{Lie_Algebras_table}
\end{center}
\end{table}

Let us choose basis $\{K_1,K_2\}$ of $\mathcal{A}$ in a way that is completely arbitrary for now, and let $(q,p,x,y)$ be corresponding coordinates introduced in Lemma \ref{lemat_o_drugim_Killingu}. One can rewrite the formula (\ref{ogolna_postac_wektora_Killinga}) in the form
\begin{equation}
K_{3} = f \partial_{q} + h \partial_{p} + \mathcal{K}_{p} \, \partial_{x} - \mathcal{K}_{q} \, \partial_{y}
\end{equation}
where
\begin{equation}
\mathcal{K} := fy - hx - \frac{1}{\Lambda} \delta, \ \delta^{\dot{1}} =: f (q,p), \ \delta^{\dot{2}} =: h(q,p), \ \delta = f_{q} + h_{p}
\end{equation}
The master equation reads
\begin{eqnarray}
\label{master_eq_dla_trzeciego}
\mathcal{K}_{p} \Theta_{x} - \mathcal{K}_{q} \Theta_{y} + 2 \delta \Theta &=& \frac{1}{6} \frac{\partial^{2} h}{\partial q^{\dot{B}} \partial q^{\dot{C}}} \, x p^{\dot{B}} p^{\dot{C}} - \frac{1}{6} \frac{\partial^{2} f}{\partial q^{\dot{B}} \partial q^{\dot{C}}} \, y p^{\dot{B}} p^{\dot{C}}
\\ \nonumber
&& +\frac{1}{2 \Lambda} \frac{\partial^{2} \delta}{\partial q^{\dot{A}} \partial q^{\dot{B}}} \, p^{\dot{A}} p^{\dot{B}} + \zeta_{\dot{A}} p^{\dot{A}} - \frac{1}{\Lambda} \frac{\partial \zeta^{\dot{A}}}{\partial q^{\dot{A}}}
\end{eqnarray}
and the commutation rules are
\begin{eqnarray}
\label{komutatory_dla_D3}
[K_1, K_3] &=& f_{q} \, K_1 + h_{q} \, K_2 + \mathcal{K}_{qp} \, \partial_{x} - \mathcal{K}_{qq} \, \partial_{y}
\\ \nonumber
[K_2, K_3] &=& f_{p} \, K_1 + h_{p} \, K_2 + \mathcal{K}_{pp} \, \partial_{x} - \mathcal{K}_{qp} \, \partial_{y}
\end{eqnarray}

\begin{Lemat}
\label{lemat_o_trzecim_Killingu}
The vector $K_{3}$ cannot commute neither with $K_1$ nor with $K_2$, otherwise $\mathcal{D}=0$.
\end{Lemat}
\begin{proof}
Assume that $[K_1, K_3]=0$. This implies $f=f(p)$, $h=h(p)$ and $\delta = \delta(p)$. Differentiating Eq. (\ref{master_eq_dla_trzeciego}) with respect to $x$ twice one obtains
\begin{equation}
\left( f_{p} y - h_{p} x - \frac{1}{\Lambda} h_{pp} \right) \Theta_{xxx} = 0
\end{equation}
Hence, $f_{p} = h_{p}=0$ or $\Theta_{xxx}=0$. The first possibility gives $K_{3} = f_{0} K_{1} + h_{0} K_2$ so $K_{3}$ is a linear combination of $K_1$ and $K_2$. The second possibility leads to an algebraic degeneracy of the ASD Weyl tensor, by the virtue of Corollary \ref{wniosek_o_niezerowosci}.

The case $[K_2, K_3]=0$ can be treated analogously.
\end{proof}

\begin{Uwaga}
\label{uwaga_o_algebrach}
\normalfont
Note, that an algebra with $[K_1, K_2]=0$, $[K_1, K_3]=0$ and $[K_2, K_3] \ne 0$ (or $[K_1, K_2]=0$, $[K_2, K_3]=0$ and $[K_1, K_3] \ne 0$) is $A_{2,1} \oplus A_{1}$ while an algebra with  $[K_1, K_2]=[K_1, K_3]=[K_2, K_3] = 0$ is $3A_{1}$. Hence, both these algebras lead to an algebraically degenerate ASD Weyl tensor. 
\end{Uwaga}

The key function $\Theta$ depends on variables $x$ and $y$, while functions $f$ and $h$ on $q$ and $p$. Consistency conditions between $\Theta(x,y)$, $f(q,p)$ and $h(q,p)$ should be established via Eq. (\ref{master_eq_dla_trzeciego}). However, this approach turns out to be quite problematic. Thus, we find the form of $K_{3}$ in a slightly different way. Values of the commutators (\ref{komutatory_dla_D3}) are spanned by $\{ K_1, K_2, K_3 \}$. Hence, there must exist constants $\alpha$ and $\beta$ such that
\begin{equation}
\label{proportionality}
\mathcal{K}_{pq} = \alpha \mathcal{K}_{p}, \ \mathcal{K}_{qq} = \alpha \mathcal{K}_{q}, \ 
\mathcal{K}_{pp} = \beta \mathcal{K}_{p},  \ \mathcal{K}_{qp} = \beta \mathcal{K}_{q}
\end{equation}
With relations (\ref{proportionality}) assumed, the commutators (\ref{komutatory_dla_D3}) take the form
\begin{eqnarray}
[K_1, K_3] &=& (f_{q} - \alpha f) K_1 + (h_{q} - \alpha h) K_2 + \alpha K_3
\\ \nonumber
[K_2, K_3] &=& (f_{p} - \beta f) K_1 + (h_{p} - \beta h) K_2 + \beta K_3
\end{eqnarray}
Let us consider possible values of the proportionality factors $\alpha$ and $\beta$. 

\textbf{The case $\alpha=\beta=0$}. Hence, $\mathcal{K}_{qq} = \mathcal{K}_{qp} = \mathcal{K}_{pp}=0$ implying $f_{qq} = f_{qp} = f_{pp} = h_{qq} = h_{qp} = h_{pp}=0$. Thus, $f=a_{0} p + b_{0} q + f_{0}$, $h=m_{0} p + n_{0} q + h_{0}$ where $a_{0}$, $b_{0}$, $m_{0}$, $n_{0}$, $f_{0}$ and $h_{0}$ are some constants. Both $f_{0}$ and $h_{0}$ can be brought to $0$ without any loss of generality. Finally, one finds
\begin{equation}
\label{wektor_K3}
K_{3} = a_{0} (p \partial_{q} + y \partial_{x}) + b_{0} (q \partial_{q} - y \partial_{y}) + n_{0} (q \partial_{p} + x \partial_{y}) + m_{0} (p \partial_{p} - x \partial_{x})
\end{equation}
and the master equation (\ref{master_eq_dla_trzeciego}) takes the form
\begin{equation}
\label{master_dla_K3}
(a_{0} y - m_{0} x) \Theta_{x} + (n_{0} x - b_{0}y) \Theta_{y} + 2 (b_{0} + m_{0}) \Theta = \zeta_{\dot{1}} x + \zeta_{\dot{2}} y
\end{equation}
where $\zeta_{\dot{1}}$ and $\zeta_{\dot{2}}$ are constants. Also, the commutators read
\begin{eqnarray}
\label{kommutatory_dwa}
[K_1, K_3] &=& b_{0} \, K_1 + n_{0} \, K_2 
\\ \nonumber
[K_2, K_3] &=& a_{0} \, K_1 + m_{0} \, K_2 
\end{eqnarray}

\textbf{The case $\alpha=0$, $\beta \ne 0$}. We easily find $\mathcal{K}_{q}=0 \ \Longrightarrow \ f_{q}=h_{q}=0$. This yields $[K_1, K_3] =0$. By the virtue of the Lemma \ref{lemat_o_trzecim_Killingu} this case leads to the algebraic degeneracy of an ASD Weyl tensor. \textbf{The case $\alpha\neq0$, $\beta = 0$} can be analyzed in a completely analogous way.

\textbf{The case $\alpha \ne 0$, $\beta \ne 0$} corresponds to an algebra $A_{2,1} \oplus A_{1}$ and by the virtue of the Remark \ref{uwaga_o_algebrach} it also leads to an algebraically degenerate ASD Weyl tensor. We omit tedious but straightforward calculations that prove this observation.

We conclude analysis of possible values of proportionality factors $\alpha$, $\beta$ with the observation that the only viable Lie algebras are given by formulas (\ref{wektor_K3}) and (\ref{kommutatory_dwa}). However, the choice of a basis $\{ K_1, K_2 \}$ of trivial subalgebra $\mathcal{A}$ has been not fixed yet. Using this freedom, together with possible scaling of $K_3$, one can bring constants $b_0,n_0,a_0,m_0$ to standard values summarized in Table \ref{Lie_Algebras_table_2}.
\begin{table}[ht]
\begin{center}
\begin{tabular}{|c|c|c|c|c|}   \hline
Algebra  &   $b_{0}$ & $n_{0}$  &  $a_{0}$  & $m_{0}$           \\ \hline \hline
 $A_{3,2}$ &  1 & 0  &  1  & 1    \\ \hline
 $A_{3,3}$ &  1 & 0  &  0  & 1 \\ \hline
 $A_{3,4}$ &  1 & 0  &  0  & -1 \\ \hline
 $A_{3,5}^{m_{0}}$ & 1 & 0  &  0  & $m_{0}$ \\ \hline
 $A_{3,6}$ & 0 & -1 &  1  & 0 \\ \hline
 $A_{3,7}^{\alpha_0}$ & $\alpha_0$ & -1  &  1  & $\alpha_0$ \\ \hline
\end{tabular}
\caption{Possible 3D algebras of symmetries for gpKE-spaces, with a trivial 2D subalgebra. The parameters $m_0$ and $\alpha_0$ are subject to the restrictions $0<|m_0|<1$ and $\alpha_0 >0$.}
\label{Lie_Algebras_table_2}
\end{center}
\end{table}

Let us briefly describe the procedure of arriving at these standard values. Using spinorial notation a basis of the subalgebra  $\mathcal{A}$ can be written as $K_{\dot{A}}=\frac{\partial}{\partial q^{\dot{A}}}$. Relations (\ref{wektor_K3}) and (\ref{kommutatory_dwa}) take form
\begin{equation}
K_3 = q^{\dot{A}} \tensor{M}{_{\dot{A}}^{\dot{B}}}\frac{\partial}{\partial q^{\dot{B}}} - p_{\dot{B}} \tensor{M}{_{\dot{A}}^{\dot{B}}}\frac{\partial}{\partial p_{\dot{A}}} 
\end{equation}
and
\begin{equation}
    [K_{\dot{A}},K_3]=\tensor{M}{_{\dot{A}}^{\dot{B}}} K_{\dot{B}}
\end{equation}
with a matrix
\begin{equation}
\begin{pmatrix}
\tensor{M}{_{\dot{A}}^{\dot{B}}}
\end{pmatrix}=
\begin{pmatrix}
b_0 & n_0\\
a_0 & m_0
\end{pmatrix}
\end{equation}
Linear transformation $K'_{\dot{A}}= \tensor{S}{_{\dot{A}}^{\dot{B}}} K_{\dot{B}}$ of a basis of  $\mathcal{A}$ results in the similarity transformation $\tensor{{M'}}{_{\dot{A}}^{\dot{B}}} = \tensor{{S}}{_{\dot{A}}^{\dot{C}}}
\tensor{{M}}{_{\dot{C}}^{\dot{D}}} \tensor{{S^{-1}}}{_{\dot{D}}^{\dot{B}}}$. Such a transformation can be used to bring $\tensor{{M}}{_{\dot{A}}^{\dot{B}}}$ to a real canonical form first. Then redundant parameters can be removed by a scaling  of $K'_{\dot{A}}$ and $K_3$, and in turn values from the Table \ref{Lie_Algebras_table_2} are obtained (with a trivial reordering of basis of $\mathcal{A}$ for the case of  $A_{3,2}$). If $\alpha<0$ then it is enough to reverse the vectors $K_1$ and $K_3$ to get $\alpha_0>0$. The restriction on $m_0$ will be justified at the beginning of Section \ref{Algebra_A35gen}. Thus, without loss of generality, we assume from now that Killing vectors $K_1$, $K_2$ and $K_3$ are chosen to produce commutation relations described in the Table \ref{Lie_Algebras_table_2}.

Finally let us remark that transformation formulas for the constants $\zeta_{\dot{A}}$ will be helpful in a further analysis. With relations (\ref{gaug_pozostaly_po_drugim_Killingu}) assumed, we find the following:
\begin{eqnarray}
\label{transformacje_na_stale_xi}
\zeta'_{\dot{1}} &=& \zeta_{\dot{1}} - n_{0} H_{\dot{2}} - (2b_{0} + m_{0}) H_{\dot{1}}
\\ \nonumber
\zeta'_{\dot{2}} &=& \zeta_{\dot{2}} - a_{0} H_{\dot{1}} - (2m_{0} + b_{0}) H_{\dot{2}}
\end{eqnarray}


\section{Metrics}
\label{Sekcja_4}
\setcounter{equation}{0}
In this Section we advance the process of solving $\mathcal{HH}$-equation as far as possible for all cases listed in the Table \ref{Lie_Algebras_table_2}. The corresponding first-order ordinary differential equations are obtained, which, in most cases, we are unable to solve. However, this does not prevent us from establishing existence of gpKE-spaces with a 3D algebra of Killing vectors of a given type, due to the following 
\begin{Uwaga}
\normalfont
\label{uwaga_o_istnieniu}
Suppose that $\mathcal{HH}$-equation has been reduced to a first-order ODE
\begin{equation}
\label{bardzo_ogolna_postac_zred_HH}
    U_\tau = \varphi(\tau,U)
\end{equation}
and that the discriminant $\mathcal{D}$ has been, using relation between $\Theta$ and $U$ and by virtue of equation (\ref{bardzo_ogolna_postac_zred_HH}), expressed as a function $\mathcal{D}(\tau,U)$ which is not identically equal to zero. 
Then the existence of a solution satisfying the condition of algebraic nondegeneracy can be obtained from the Peano theorem by taking as an initial condition any point $(\tau_0,U_0)$ such that in its neighborhood the functions $\varphi(\tau,U)$, $\mathcal{D}(\tau,U)$ are continuous and $\mathcal{D}(\tau,U) \neq  0$. 
\end{Uwaga}

Applying this Remark requires a computation of the discriminant $\mathcal{D}$. This, in most cases, has been accomplished by means of a computer algebra software. If overly complicated, the resulting expression has been abbreviated to a more concise form, preserving its general structure.

\subsection{Algebra $A_{3,2}$}
\label{algebra_A32}

Let us consider an algebra $A_{3,2}$. Hence, structure constants read $a_{0}=b_{0}=m_{0}=1$, $n_{0}=0$ (compare Table \ref{Lie_Algebras_table_2}).

\begin{Twierdzenie}
\label{Twierdzenie_A_32}
Let $(\mathcal{M}, ds^{2})$ be an Einstein complex space of the type $[\textrm{D}]^{nn} \otimes [\textrm{I}]$ equipped with a 3D algebra of infinitesimal symmetries $A_{3,2}$. Then there exists a local coordinate system $(q,p,w,y)$ such that the metric takes the form 
\begin{eqnarray}
\label{metryka_wspol_qpzy}
\frac{1}{2} ds^{2} &=& dqdy - dx dp - \frac{\Lambda}{3} y^{2} \left(  \frac{3 \Sigma}{ 12 \Sigma -(12 w +1) ( 3w + 1  ) } - 1 \right) dq^{2}
\\ \nonumber
&& - \frac{2 \Lambda}{3} y \left(  3 w y +  (y-x) \frac{3 \Sigma}{ 12 \Sigma -(12 w +1) ( 3w + 1  ) } +  x \right)  dpdq
\\ \nonumber
&& - \frac{\Lambda}{3}  \bigg( (12 \Sigma - 36 w^{2} - 16w)y^{2}  +  ( 7y^{2} - 6 xy ) w 
\\ \nonumber
&& \ \ \ \ \ \ \ \ \ \ \ \ + ( y-x )^{2} \frac{3 \Sigma}{ 12 \Sigma -(12 w +1) ( 3w + 1  ) } -  x^2 \bigg)  dp^{2}
\end{eqnarray}
where $\Lambda \ne 0$ is the cosmological constant, $\Sigma = \Sigma (w)$ is a holomorphic function which satisfies the equation
\begin{equation}
\label{zredukowane_HH_A3D_wersja_4}
\Sigma \Sigma_{w} = \left( 10w + \frac{4}{3} \right) \Sigma - \frac{1}{3} w (12 w +1) ( 3w + 1  )
\end{equation}
and
\begin{equation}
\label{relation_xyw_A_32}
 x (w,y) = -y  \ln y + y \left( 4w - \int \frac{(12w+1)(3w+1)}{3\Sigma} dw  \right)
\end{equation}
\end{Twierdzenie}
\begin{proof}
If $a_{0}=b_{0}=m_{0}=1$, $n_{0}=0$, both constants $\zeta_{\dot{1}}$ and $\zeta_{\dot{2}}$ can be gauged away (compare the transformation formulas (\ref{transformacje_na_stale_xi})) and the general solution to the master equation (\ref{master_dla_K3}) yields
\begin{equation}
\label{A_32_pomocnicze_4}
\Theta = \frac{\Lambda}{3} y^{4} \, F(z), \ \ z:= \ln y + \frac{x}{y}
\end{equation}
The second derivatives of the key function read 
\begin{eqnarray}
\label{A_32_pomocnicze_2}
&& \Theta_{xx} = \frac{\Lambda}{3}  y^{2} F_{zz}, \ \Theta_{xy} =\frac{\Lambda}{3}  \left( 3y^{2} F_{z} + y (y-x) F_{zz} \right), 
\\ \nonumber
&& \Theta_{yy} = \frac{\Lambda}{3}  \left( 12y^{2} F + y(7y - 6x) F_{z} + (y-x)^{2} F_{zz} \right)
\end{eqnarray}
Hence, the $\mathcal{HH}$-equation (\ref{rownanie_hiperniebianskie_dwa_komutujace_Killingi}) takes the form
\begin{equation}
\label{zredukowane_HH_A3D}
F_{z} F_{zz} + 12 F F_{zz} - 9 F_{z}^{2} -    F_{zz} - 4 F_{z} - 3F = 0
\end{equation}
Eq. (\ref{zredukowane_HH_A3D}) can be reduced to a first--order ODE. Indeed, the substitution $F_{z} = w (F)$ implies $F_{zz} = w w_{F}$. Consequently, from (\ref{zredukowane_HH_A3D}) one obtains
\begin{equation}
\label{zredukowane_HH_A3D_wersja_2}
(w^{2}  -w + 12 F w)w_{F} - 9 w^{2}  - 4 w - 3F = 0
\end{equation}
The discriminant $\mathcal{D}$ for a key function obtained from the equation (\ref{zredukowane_HH_A3D_wersja_2}) is given by 
\begin{equation}
\label{wyroznik_A32}
\begin{split}
    \mathcal{D}(F,w)=& \frac{64 \Lambda ^6 (3 F+w)^2 }{(12 F+w-1)^{16}} \Big(18144 F^2 w^2+10368 F^3 w+1998 F^2 w+1512
   F^3\\
   &-135 F^2+1944 F w^4+7668 F w^3-225 F w^2-261 F w-810 w^5\\
   &-534 w^4-727
   w^3-126 w^2\Big)^2
   \Big(-979776 F^2 w^5-1119744 F^3 w^4\\
   &+23112 F^2w^4+2066688 F^3 w^3
   -460440 F^2 w^3+6531840 F^4 w^2\\
   &-1186272 F^3 w^2+45333
   F^2 w^2+7464960 F^5 w-1425600 F^4 w\\
   &+84348 F^3 w-3294 F^2 w+2985984
   F^6-622080 F^5+46980 F^4\\
   &-2268 F^3+81 F^2-279936 F w^6-84132 F w^5-72672 F
   w^4\\
   &+14898 F w^3-144 F w^2+162 F w+26244 w^8+26568 w^7+35656 w^6\\
   &+17480
   w^5+7333 w^4+882 w^3+81 w^2\Big)
\end{split}
\end{equation}
The existence of a considered space follows from the Remark \ref{uwaga_o_istnieniu} applied for relations (\ref{zredukowane_HH_A3D_wersja_2}) and (\ref{wyroznik_A32}).
Now let us treat $w$ as a variable and $F$ as a function of $w$, $F=F(w)$. Note that $1 = F_{w} w_{F}$. Hence, from (\ref{zredukowane_HH_A3D_wersja_2}) it follows that
\begin{equation}
\label{zredukowane_HH_A3D_wersja_3}
(9w^{2} + 4w +3F) F_{w} = w^{2} - w +12wF
\end{equation}
Another substitution 
\begin{equation}
\label{A_32_pomocnicze_3}
3 \Sigma := 3F + 9w^{2} + 4w
\end{equation}
reduces Eq. (\ref{zredukowane_HH_A3D_wersja_3}) to the form (\ref{zredukowane_HH_A3D_wersja_4}).

One easily finds how $F$, $F_{z}$ and $F_{zz}$ are expressed in terms of $\Sigma$ and $w$:
\begin{equation}
\label{A_32_pomocnicze_1}
F = \Sigma - \frac{4}{3}w - 3 w^{2}, \ F_{z} = w, \ F_{zz} = \frac{3 \Sigma}{ 12 \Sigma -(12 w +1) ( 3w + 1  ) }
\end{equation}
Inserting (\ref{A_32_pomocnicze_1}) into (\ref{A_32_pomocnicze_2}) and then into (\ref{metryka_w_xyqp}) we arrive at the metric (\ref{metryka_wspol_qpzy}).

Finally, from the relation $F_{z} = w (F)$ one finds that $z = \int \frac{dF}{w} = \int \frac{F_{w} dw}{w}$. Using the relation (\ref{A_32_pomocnicze_3}), Eq. (\ref{zredukowane_HH_A3D_wersja_4}) and the definition of $z$ (\ref{A_32_pomocnicze_4}) one proves that (\ref{relation_xyw_A_32}) is correct.
\end{proof}

\begin{Uwaga}
\normalfont
To bring Eq. (\ref{zredukowane_HH_A3D_wersja_4}) to the canonical form one defines a new variable $v$ as follows
\begin{equation}
v := 5w^{2} + \frac{4}{3} w
\end{equation}
Hence, Eq. (\ref{zredukowane_HH_A3D_wersja_4}) takes the form 
\begin{equation}
\label{postac_kanoniczna_drugirodzaj_A32}
\Sigma \Sigma_{v} - \Sigma = - \frac{(3V-4)(2V-1)(V + 2)}{500 V}, \ V (v) := \pm \sqrt{\frac{16}{9} + 20v}
\end{equation}
where $\Sigma = \Sigma (v)$.
\end{Uwaga}

\subsection{Algebra $A_{3,3}$}
\label{algebra_A33}

In this case $b_{0}=1$, $n_{0} = a_{0} = 0$ and $m_{0} = 1$. The constants $\zeta_{\dot{A}}$ can be gauged away (compare (\ref{transformacje_na_stale_xi})) and the solution to Eq. (\ref{master_dla_A_33}) reads
\begin{equation}
\Theta = y^{4} F(z), \ z: =\frac{x}{y}
\end{equation}
The $\mathcal{HH}$-equation (\ref{rownanie_hiperniebianskie_dwa_komutujace_Killingi}) takes the form
\begin{equation}
12FF_{zz} - 9F_{z}^{2} - \Lambda F=0
\end{equation}
which can be simply solved. Finally, the solution for $\Theta$ reads
\begin{equation}
\label{Rozwiazanie_A33_for_m1}
\Theta = \left(    F_{0}  (x+ G_{0} y)^{2} + \frac{\Lambda y^{2}}{48F_{0}}  \right)^{2}, \ F_{0}, G_{0} \ \textrm{are constants}
\end{equation}
However, the solution (\ref{Rozwiazanie_A33_for_m1}) implies $\mathcal{D}=0$. Hence, the algebra $A_{3,3}$ corresponds to algebraically degenerate solutions.

\subsection{Algebra $A_{3,4}$}
\label{Algebra_A34}

Let us deal with an algebra $A_{3,4}$. Hence, we put $b_{0}=1$, $n_{0} = a_{0} = 0$ and $m_{0} =- 1$.

\begin{Twierdzenie}
\label{Twierdzenie_A_34}
Let $(\mathcal{M}, ds^{2})$ be an Einstein complex space of the type $[\textrm{D}]^{nn} \otimes [\textrm{I}]$ equipped with a 3D algebra of infinitesimal symmetries $A_{3,4}$. Then there exists a local coordinate system $(q,p,x,y)$ such that the metric takes the form 
\begin{eqnarray}
\label{metryka_kanoniczna_A_34}
\frac{1}{2} ds^{2}&=& dq dy - dp dx -  \frac{\Lambda}{2^{\frac{4}{3}}} r^{-\frac{4}{3}} \left( \frac{3}{4} - \frac{\Sigma}{2r } \right)^{-1}  (y dq + x dp)^{2}
\\ \nonumber
 && \ \ \ \ \ \ \ \ \ \ \ \ \ \ \ 
+ \frac{\Lambda}{3} \left( \frac{3}{4} - \frac{\Sigma}{2r } \right) (y dq - x dp)^{2}
 \end{eqnarray}
where $\Lambda \ne 0$ is the cosmological constant, $\Sigma = \Sigma (r)$ is a holomorphic function which satisfies the equation
\begin{equation}
\label{HH_equation_for_A_34_formalizm_Sigma}
 \Sigma \Sigma_{r} - \Sigma = \frac{3}{4}r + 2^{\frac{2}{3}} r^{-\frac{1}{3}}
\end{equation}
and
\begin{equation}
\label{algebra_A_34_zwiazek_xy_r}
xy =  r^{\frac{2}{3}} \exp \left( -  \int \frac{dr}{\Sigma} \right)
\end{equation}
\end{Twierdzenie}
\begin{proof}
Because $m_{0} = -1$ the constants $\zeta_{\dot{A}}$ can be put zero without any loss of generality. A general solution to the master equation (\ref{master_dla_A_33}) reads
\begin{equation}
\Theta = \Theta (z), \ z:= xy
\end{equation}
The $\mathcal{HH}$-equation takes the form
\begin{equation}
\label{HH_equation_for_A_34}
\Theta_{z}^{2} + 2z \Theta_{z} \Theta_{zz} + \frac{\Lambda}{3} ( 4 z^{2} \Theta_{zz} - 4z \Theta_{z} + 3\Theta ) =0
\end{equation}
If we introduce a new function $T(v)$ defined as follows
\begin{equation}
\label{Definicja_T_via_Theta}
\Theta (z) := \frac{4 \Lambda}{3} \, T(v), \ v:= z^{\frac{1}{2}}
\end{equation}
then the second derivatives of $\Theta$ yield
\begin{eqnarray}
&& \Theta_{xx} = \frac{\Lambda}{3} \frac{y^{2}}{v^{2}} \left( T_{vv} - \frac{1}{v} \, T_{v} \right), \ 
\Theta_{yy} = \frac{\Lambda}{3} \frac{x^{2}}{v^{2}} \left( T_{vv} - \frac{1}{v} \, T_{v} \right), 
\\ \nonumber
&& \Theta_{xy} = \frac{\Lambda}{3} \frac{xy}{v^{2}} \left( T_{vv} + \frac{1}{v} \, T_{v} \right)
\end{eqnarray}
and the metric expressed in terms of $T(v)$ and written in the hyperheavenly coordinates reads
\begin{equation}
\label{metryka_algebra_A_34_FORMA}
\frac{1}{2} ds^{2} = dq dy - dp dx - \frac{\Lambda}{3} \frac{T_{vv}}{v^{2}} (y dq + x dp)^{2} + \frac{\Lambda}{3} \left(\frac{T_{v}}{v^{3}} + 1 \right) (y dq - x dp)^{2}
\end{equation}
Substituting (\ref{Definicja_T_via_Theta}) into (\ref{HH_equation_for_A_34}) one gets
\begin{equation}
\label{HH_equation_for_A_34_formalizm_T}
T_{v} T_{vv} + v^{3} T_{vv} - 3 v^{2} T_{v} + 3v T = 0
\end{equation}
Equation (\ref{HH_equation_for_A_34_formalizm_T}) can be reduced to an autonomous form by the substitution $T(v)=v^4 g (w )$ where $w:=\ln v$. One obtains the relation
\begin{equation}
    3 g + 48 g^2 + 4 g_w  + 40 g g_w  + 7  g_w^2 +  g_{ww}  + 4 g g_{ww}  +  g_w g_{ww} = 0.
\end{equation}
Putting $g_w = Q(g) - 3 g$ and consequently $g_{ww} = (Q_g - 3) (Q - 3g )$ one brings it to the first--order equation
\begin{equation}
\label{HH_equation_for_A_34_formalizm_Q}
   (1 + 4 g + 4 Q) Q +  ( Q -3 g ) (1 + g +  Q) Q_g = 0 
\end{equation}
for a function $Q(g)$. Using this relation one can express the discriminant $\mathcal{D}$ in terms of variables $g$ and $Q$ 
\begin{equation}
\label{wyroznik_dla_algebry_A34_Qg}
\begin{split}
    \mathcal{D}(g,Q)=&\frac{576 \Lambda ^6 Q^2}{(g+ Q+1)^{16}} \left(g (g+1)+2 (g-1) Q+ Q^2\right)^2\\
    &\times\left(2 g (g+1)^2+\left(6 g^2+2 g-1\right) Q+2 (3 g-1)  Q^2+2  Q^3\right)^2 \\
    &\times\left(2 g (g+1)^3 +(g+1) (8 g^2+ 10 g -1) Q+3 \left(4 g^2+6 g-5\right)  Q^2\right.\\
    &\quad \quad \left. +2 (4 g+3)  Q^3+2 Q^4\right)^2
\end{split}    
\end{equation}
and the problem existence of can be resolved by means of the Remark \ref{uwaga_o_istnieniu}. Furthermore, one easily finds that
\begin{equation}
\frac{T_{v}}{v^{3}} = Q +g , \ \frac{T_{vv}}{v^{2}} = \frac{3Q}{Q +g +1}
\end{equation}
and the metric can be written as
\begin{equation}
\label{metryka_algebra_A_34_FORMA_2}
\frac{1}{2} ds^{2} = dq dy - dp dx -  \frac{\Lambda Q}{Q +g +1} (y dq + x dp)^{2} + \frac{\Lambda}{3} (Q +g +1) (y dq - x dp)^{2}
\end{equation}
The relation between $x$, $y$ and $g$ is given by the formula
\begin{equation}
\label{algebra_A_34_zwiazek_xy_g}
xy = v^{2} = e^{2w} = \exp \left( 2 \int \frac{dg}{Q-3g}  \right)
\end{equation}
Finally, one reduces Eq. (\ref{HH_equation_for_A_34_formalizm_Q}) to an Abel equation. We treat $Q$ as a variable and $g$ as a function of $Q$. Of course, $Q_{g} g_{Q} = 1$. To be more precise, we put
\begin{equation}
Q =: \frac{1}{2^{\frac{4}{3}} r^{\frac{4}{3}}}, \ g  =: - \frac{\Sigma}{2r}  - \frac{1}{2^{\frac{4}{3}} r^{\frac{4}{3}}} -\frac{1}{4}
\end{equation}
where $\Sigma (r)$ is a new function. With such a substitution, Eq. (\ref{HH_equation_for_A_34_formalizm_Q}) and the metric (\ref{metryka_algebra_A_34_FORMA_2}) reduce to the forms (\ref{HH_equation_for_A_34_formalizm_Sigma}) and (\ref{metryka_kanoniczna_A_34}), respectively. Finally, one finds that (\ref{algebra_A_34_zwiazek_xy_g}) yields
\begin{equation}
\label{algebra_A_34_zwiazek_xy_g_posredni}
xy  = \exp \left( 2^{\frac{7}{3}} \int \frac{  r^{\frac{4}{3}} g_{r} dr}{1-3 \cdot 2^{\frac{4}{3}} r^{\frac{4}{3}}g}  \right)
\end{equation}
Using the definition of $\Sigma$ and Eqs. (\ref{HH_equation_for_A_34_formalizm_Sigma}) in (\ref{algebra_A_34_zwiazek_xy_g_posredni}) one arrives at (\ref{algebra_A_34_zwiazek_xy_r}).
\end{proof}

It is probably the best, what we can do within an algebra $A_{3,4}$. Eq. (\ref{HH_equation_for_A_34_formalizm_Sigma}) seems to be simpler then its counterpart in algebra $A_{3,2}$, but still its general solution is unknown (at least, it cannot be found in an extensive list of known solutions to an Abel equation published in \cite{Polyanin}).
Formula (\ref{wyroznik_dla_algebry_A34_Qg}) implies $\mathcal{D} \geqslant 0$, thus for the nondegenerated real neutral case $\mathcal{D}\neq 0$ one must deal with types $[\textrm{I}_{r}]$ or $[\textrm{I}_{c}]$. Unfortunately, cumbersome and rather obscure expressions for $\mathcal{P}$ and $\mathcal{R}$ prevented us from separating these two options.

\subsection{Algebra $A_{3,5}^{m_{0}}$, $m_{0} \ne -\frac{1}{2}$}
\label{Algebra_A35gen}

Here we have $b_{0}=1$, $n_{0} = a_{0} = 0$, while $m_{0}$ could be, at a glance, arbitrary, except for $m_{0}=0$ which leads to an algebra $A_{2,1} \oplus A_{1}$, and, consequently, to $\mathcal{D}=0$. The master equation (\ref{master_dla_K3}) takes the form
\begin{equation}
\label{master_dla_A_33}
 - m_{0} x \Theta_{x}  - y \Theta_{y} + 2 (1 + m_{0}) \Theta = \zeta_{\dot{1}} x + \zeta_{\dot{2}} y
\end{equation}
Note, that if one replaces $m_{0}  \longleftrightarrow  \frac{1}{m_{0}}$ and $x \longleftrightarrow y$, Eq. (\ref{master_dla_A_33}) remains unchanged (with an appropriate redefinition of constants $\zeta_{\dot{A}}$). The metric (\ref{metryka_w_xyqp}) also remains unchanged under replacement $x \longleftrightarrow y$. Thus, we can restrict considerations to $m_{0}$ such that $0 < |m_{0}| < 1$. In what follows we assume also that $m_{0} \ne  -\frac{1}{2}$.

\begin{Twierdzenie}
\label{Twierdzenie_A_35}
Let $(\mathcal{M}, ds^{2})$ be an Einstein complex space of the type $[\textrm{D}]^{nn} \otimes [\textrm{I}]$ equipped with a 3D algebra of infinitesimal symmetries $A_{3,5}^{m_{0}}$ where $-1 < m_{0} < 1$, $m_{0} \ne 0$, $m_{0} \ne -\frac{1}{2}$. Then there exists a local coordinate system $(q,p,x,y)$ such that the metric takes the form 
\begin{eqnarray}
\label{metryka_kanoniczna_A_35}
\frac{1}{2} ds^{2} &=& dqdy-dpdx + \frac{\Lambda}{3} (xdp-ydq)^{2} - \frac{\Lambda}{3} (2m_{0}+1)(2m_{0} +2) \, g x^{2} \, dp^{2}
\\ \nonumber
&& - \frac{\Lambda}{3} \Bigg[ \left(  \frac{(m_{0}-1)^{\frac{8}{3}}}{2^{\frac{4}{3}} r^{\frac{4}{3}}} + \frac{3 g}{1 - m_{0}} \right) \left(-\frac{2^{\frac{2}{3}} (m_{0}-1)^{\frac{8}{3}}}{3 r^{\frac{7}{3}} g_{r}} + \frac{3 }{1 - m_{0}} \right)
\\ \nonumber
&& \ \ \ \ \ \ \ 
 - \frac{5(m_{0}-1)^{\frac{8}{3}}}{2^{\frac{4}{3}} r^{\frac{4}{3}}} - \frac{15 g}{1 - m_{0}} + 6g \Bigg] (m_{0} x dp - y dq)^{2}
\\ \nonumber
&& +\frac{\Lambda}{3} \left( \frac{(m_{0}-1)^{\frac{8}{3}}}{2^{\frac{4}{3}} r^{\frac{4}{3}}} + \frac{3 g}{1 - m_{0}} - 2g \right) 
\\ \nonumber
&& \ \ \ \ \ \ \ \ \times \big( - m_{0} (5 m_{0} + 3) x^{2} dp^{2} - 2y^{2} dq^{2} + 2 (3 m_{0} +2) xy dq dp 
\big)
\end{eqnarray}
where $\Lambda \ne 0$ is the cosmological constant, $g=g(r)$ reads
\begin{equation}
 g = \dfrac{ (m_{0}-1)^{2}  \dfrac{\Sigma}{2r} - (1-m_{0}^{3}) \dfrac{(m_{0}-1)^{\frac{8}{3}}}{2^{\frac{4}{3}} r^{\frac{4}{3}}} +\dfrac{1}{4} (m_{0}-1)^{2}} {(m_{0}+2) (2 m_{0}+1)}
\end{equation}
$\Sigma = \Sigma (r)$ is a holomorphic function which satisfies the equation
\begin{equation}
\label{Algebra_A_35_roownanie_final}
\Sigma \Sigma_{r} - \Sigma = \frac{3}{4}r - \frac{2^{-\frac{2}{3}}}{3} (m_{0}+1)^{2} (m_{0}-1)^{\frac{16}{3}} r^{-\frac{5}{3}} - 2^{\frac{2}{3}} m_{0} (m_{0}-1)^{\frac{5}{3}} r^{-\frac{1}{3}}
\end{equation}
and
\begin{equation}
\label{algebra_A_35_zwiazek_xy_r}
\frac{y^{m_{0}}}{x} = \exp \left(  \int \frac{g_{r} dr}{\frac{(m_{0}-1)^{\frac{8}{3}}}{2^{\frac{4}{3}} r^{\frac{4}{3}}}
 + \frac{3 g}{1 - m_{0}}} \right)
\end{equation}
\end{Twierdzenie}

\begin{proof}
If $m_{0} \ne  -\frac{1}{2}$ then $\zeta_{\dot{A}}$ can be gauged away and a solution to the master equation reads
\begin{equation}
\label{Theta_for_m0_different}
 \Theta = y^{2 +2 m_{0}} \, F(z), \ z = \frac{y^{m_{0}}}{x}
\end{equation}
Hence
\begin{eqnarray}
\label{Theta_for_m0_different_drugie_pochodne}
\Theta_{xx} &=& y^{2} (2z^{3} F_{z} + z^{4} F_{zz})
\\ \nonumber
\Theta_{xy} &=& -y^{1 + m_{0}} \big( (2 + 3m_{0}) z^{2} F_{z} + m_{0} z^{3} F_{zz} \big)
\\ \nonumber
\Theta_{yy} &=& y^{2m_{0}} \big( (2 + 2m_{0}) (2 m_{0}+1) F + m_{0} (5m_{0} +3) zF_{z} + m_{0}^{2} z^{2} F_{zz} \big)
\end{eqnarray}
Inserting (\ref{Theta_for_m0_different}) and (\ref{Theta_for_m0_different_drugie_pochodne}) into the $\mathcal{HH}$-equation (\ref{rownanie_hiperniebianskie_dwa_komutujace_Killingi}) we obtain
\begin{eqnarray}
\label{H_for_A33}
\nonumber
&& z^{3} \big[  m_{0} (m_{0}-1) z^{2} F_{z} F_{zz} + (2m_{0}+1)(2m_{0} +2) (zFF_{zz} + 2FF_{z}) + (m_{0}^{2} - 6m_{0} -4) z F_{z}^{2} \big]
\\ 
&&+ \frac{\Lambda}{3} \big[ (1-2m_{0})(1 + 2 m_{0}) F - (m_{0}-1) (5m_{0}-1) zF_{z} - (m_{0}-1)^{2} z^{2} F_{zz} \big] =0
\end{eqnarray}
The substitution 
\begin{equation}
F(z) =: \frac{\Lambda}{3} z^{-2} g(w), \ w:= \ln z
\end{equation}
brings Eq. (\ref{H_for_A33}) to an autonomous form 
\begin{eqnarray}
\label{autonomous_form_algebra_A_35}
&& m_{0} (m_{0}-1) g_{w} g_{ww} + 2 (m_{0}^{2} + 4 m_{0} +1 ) g g_{ww} - 10  (m_{0}-1) g g_{w}
\\ \nonumber
&& \ \ \ \ \ \ -(4m_{0}^{2} + m_{0} +4) g_{w}^{2} - (m_{0}-1)^{2} g_{ww} - 4 (m_{0}-1) g_{w} -12 g^{2} - 3g = 0
\end{eqnarray}
Eq. (\ref{autonomous_form_algebra_A_35}) can be reduced to a first--order ODE by a standard trick. We treat $g$ as a variable and we put $g_{w} = Q(g) + \frac{3 g}{1 - m_{0}}$ so $g_{ww} = \left(  Q + \frac{3 g}{1 - m_{0}} \right) \left(  Q_{g} + \frac{3 }{1 - m_{0}} \right)$. Finally
\begin{eqnarray}
\label{autonomous_form_algebra_A_35_first_ODE}
  &&  \left(2 g \left(m_0^3-1\right)-(m_0-1)^3\right) Q Q_g+4\left(1-
   m_0^3\right) Q^2+(m_0-1)^2 m_0 Q^2 Q_g
   \\ \nonumber
 &&  +3 g \left((m_0-1)^2-g
   (m_0+2) (2 m_0+1)\right) Q_g+\left(4 g (m_0+2) (2
   m_0+1)-(m_0-1)^2\right) Q = 0
\end{eqnarray}
Using this equation one can write down the discriminant $\mathcal{D}$ in terms of $g$ and $Q$ in order to apply Remark \ref{uwaga_o_istnieniu}. The result reads
\begin{equation}
    \mathcal{D}(g,Q)=\frac{64 \Lambda ^6 m_0^2 \, P_{3,5}}{27 (m_0-1)^{12} \big(g (m_0+2) (2m_0+1)+(m_0-1) (m_0 (Q-1)+1)\big)^{18}}
\end{equation}
where
\begin{multline}
    P_{3,5}=27 (m_0-1)^{34} (m_0+2)^4 (2 m_0+1)^4 Q^2 (2 g-m_0 Q+Q)^4 \\
    \times\big(3 g (m_0+2) (2 m_0+1)-2
   (m_0-1) (m_0 (m_0+5)+1) Q\big)^2+\dots
\end{multline}
is a complicated polynomial whose total degree in $g$ and $Q$ is $22$. The metric written in terms of $Q(g)$ takes the form
\begin{eqnarray}
\label{the_metricc_algebra_A_35_final}
\frac{1}{2} ds^{2} &=& dqdy-dpdx + \frac{\Lambda}{3} (xdp-ydq)^{2} - \frac{\Lambda}{3} (2m_{0}+1)(2m_{0} +2) \, g x^{2} \, dp^{2}
\\ \nonumber
&& - \frac{\Lambda}{3} \left( \left(  Q + \frac{3 g}{1 - m_{0}} \right) \left(  Q_{g} + \frac{3 }{1 - m_{0}} \right) - 5Q - \frac{15 g}{1 - m_{0}} + 6g \right) (m_{0} x dp - y dq)^{2}
\\ \nonumber
&& +\frac{\Lambda}{3} \left( Q + \frac{3 g}{1 - m_{0}} - 2g \right) \big( - m_{0} (5 m_{0} + 3) x^{2} dp^{2} - 2y^{2} dq^{2} + 2 (3 m_{0} +2) xy dq dp 
\big)
\end{eqnarray}
where $Q (g)$ satisfies Eq. (\ref{autonomous_form_algebra_A_35_first_ODE}) and the relation between $x$, $y$ and $g$ reads
\begin{equation}
\label{algebra_A_35_zwiazek_miedzy_xy_r}
\frac{y^{m_{0}}}{x} = \exp \left(  \int \frac{dg}{Q + \frac{3 g}{1 - m_{0}}} \right)
\end{equation}
The next step is to bring Eq. (\ref{autonomous_form_algebra_A_35_first_ODE}) to an Abel form. First, we treat $Q$ as a new variable which we denote by $s$, while $g$ is now a function of $s$, $g = g(s)$. Note, that $Q_{g} g_{s} =1$. We use also the substitution
\begin{equation}
\label{algebra_A_35_dowolne_m}
(m_{0}+2) (2 m_{0}+1) g =:s^{\frac{3}{4}} \Sigma - (1-m_{0}^{3}) s +\frac{1}{4} (m_{0}-1)^{2}
\end{equation}
where $\Sigma (s)$ is a new function. We arrive at the equation
\begin{eqnarray}
\label{algebra_A_35_dowolne_m_rownanie_A}
\Sigma \Sigma_{s} &=& - \frac{3}{8} (m_{0}-1)^{2} s^{-\frac{7}{4}} \Sigma 
\\ \nonumber
&& + \frac{1}{4} (m_{0}-1)^{3} \left( (m_{0}-1)(m_{0} + 1)^{2} s^{-\frac{1}{2}} + 3 m_{0} s^{-\frac{3}{2}} - \frac{9}{16} (m_{0}-1) s^{-\frac{5}{2}} \right)
\end{eqnarray}
Finally, the substitution 
\begin{equation}
\label{Substitution_algebra_A_35_final}
s =: \frac{(m_{0}-1)^{\frac{8}{3}}}{2^{\frac{4}{3}} r^{\frac{4}{3}}}
\end{equation}
reduces Eq. (\ref{autonomous_form_algebra_A_35_first_ODE}) to the form (\ref{Algebra_A_35_roownanie_final}).

Inserting (\ref{algebra_A_35_dowolne_m}) and (\ref{Substitution_algebra_A_35_final}) into (\ref{the_metricc_algebra_A_35_final}) and (\ref{algebra_A_35_zwiazek_miedzy_xy_r}) one proves  formulas (\ref{metryka_kanoniczna_A_35}) and (\ref{algebra_A_35_zwiazek_xy_r}).
\end{proof}

\subsection{Algebra $A_{3,5}^{m_{0}}$, $m_{0} =  -\frac{1}{2}$}
\label{Algebra_A_35_05_Sekcja}

An algebra $A_{3,5}^{m_{0}}$ with $m_{0} = - \frac{1}{2}$ must be considered separately, because a subtle difference appears in this case. Indeed, the constant $\zeta_{\dot{2}}$ cannot be gauged away anymore (compare the transformation formulas (\ref{transformacje_na_stale_xi})).

\begin{Twierdzenie}
\label{Twierdzenie_A_35_05}
Let $(\mathcal{M}, ds^{2})$ be an Einstein complex space of the type $[\textrm{D}]^{nn} \otimes [\textrm{I}]$ equipped with a 3D algebra of infinitesimal symmetries $A_{3,5}^{-\frac{1}{2}}$. Then there exists a local coordinate system $(q,p,x,w)$ such that the metric takes the form 
\begin{eqnarray}
\label{metryka_we_wspolrzednych_qpxw}
\frac{1}{2} ds^{2} &=& dq \left( \frac{Z_{w}}{x^{2}} \, dw - \frac{2Z}{x^{3}} \, dx \right) -dxdp - x^{2} \left( 2 w + \frac{Z}{Z_{w}} - \frac{\Lambda}{3} + \frac{\zeta_{0}}{Z} \right) dp^{2}
\\ \nonumber
&& - \frac{2Z}{x} \left( 4 w + \frac{2Z}{Z_{w}} + \frac{\Lambda}{3} \right) dqdp - \frac{Z^{2}}{x^{4}} \left( 2 w + \frac{4Z}{Z_{w}} - \frac{\Lambda}{3} \right) dq^{2}
\end{eqnarray}
where $\Lambda \ne 0$ is the cosmological constant, $\zeta_{0}$ is a constant and $Z=Z(w)$ reads
\begin{equation}
\label{ogolne_podstawienie_Z}
Z(w) = w^{-3} (12 w + \Lambda)^{\frac{5}{2}} \, \Sigma(w) + \frac{ 2 \zeta_{0} \left( w + \frac{\Lambda}{3}  \right)}{12 w^{2} + \Lambda w}
\end{equation}
where $\Sigma = \Sigma (w)$ is a holomorphic function which satisfies the equation
\begin{equation}
\label{master_A_35_1_2_wersja_3}
 \Sigma \Sigma_{w} = f_{1} (w) \, \Sigma + f_{0}(w)
\end{equation}
with
\begin{equation}
\label{Definicje_f0_f1_dla_algebray_A35_05}
 f_{1}(w) := \frac{2 \zeta_{0} w \left( 24w^{2} - \frac{5 \Lambda^{2}}{3}  \right)}{(12w + \Lambda)^{\frac{9}{2}}}, \ f_{0} (w) := \frac{12 \zeta_{0}^{2} \, w^{3} \left(w + \frac{\Lambda}{3} \right) \left(w - \frac{\Lambda}{3} \right) (6w + \Lambda)}{(12w + \Lambda)^{8}}
\end{equation}
\end{Twierdzenie}

\begin{proof}
Without any loss of generality one puts $\zeta_{\dot{1}} = 0$ and $\zeta_{\dot{2}} = : -\zeta_{0} $. A general solution to the master equation reads
\begin{equation}
\label{Theta_for_m0_minus_12}
 \Theta = y \, F(z) + \zeta_{0} \, y \ln y , \ z= y x^{2}
\end{equation}
where $F=F(z)$ is an arbitrary function. The second derivatives of $\Theta$ with respect to $x$ and $y$ read
\begin{equation}
\Theta_{xx} = y^{2} (2 \Omega + 4z \Omega_{z}), \ \Theta_{xy} = xy (4 \Omega +2 z \Omega_{z}), \ \Theta_{yy} = x^{2} \left( 2 \Omega + z \Omega_{z} + \frac{\zeta_{0}}{z} \right)
\end{equation}
where we denoted $\Omega (z) := F_{z}$. In coordinates $(q,p,x,z)$ the metric takes the form
\begin{eqnarray}
\label{Algebra_A_35_05_pomocnicze_1}
\frac{1}{2} ds^{2} &=& dq \left( \frac{1}{x^{2}} \, dz - \frac{2z}{x^{3}} \, dx \right) -dxdp - x^{2} \left( 2 \Omega + z \Omega_{z} - \frac{\Lambda}{3} + \frac{\zeta_{0}}{z} \right) dp^{2}
\\ \nonumber
&& - \frac{2z}{x} \left( 4 \Omega + 2z \Omega_{z} + \frac{\Lambda}{3} \right) dqdp - \frac{z^{2}}{x^{4}} \left( 2 \Omega + 4z \Omega_{z} - \frac{\Lambda}{3} \right) dq^{2}
\end{eqnarray}
The $\mathcal{HH}$-equation yields
\begin{equation}
\label{master_A_35_1_2}
6z^{2} \Omega \Omega_{z} + 12 z \Omega^{2} - 2 \zeta_{0} (\Omega + 2z \Omega_{z}) + \Lambda \left( 3z^{2} \Omega_{z} + z \Omega - \frac{2}{3} \zeta_{0} \right) =0
\end{equation}
Spaces resulting from this equation can be algebraically nondegenerate as the discriminant $\mathcal{D}$ in variables $z$ and $\Omega$ reads
\begin{equation}
    \mathcal{D}(z,\Omega)=\frac{2359296 \Lambda ^2 (\zeta_{0} +3 z\Omega)^2}{(-4 \zeta_{0} +3\Lambda  z+6 z\Omega)^{16}} \widetilde{P}_{3,5}
\end{equation}
where
\begin{equation}
 \widetilde{P}_{3,5} = 67108864 \Lambda^4 \zeta_{0}^{14}+\dots +4782969000000  \Lambda^{4} z^{14} \Omega^{14}
\end{equation}
is a polynomial in $z$ and $\Omega$. To reduce Eq. (\ref{master_A_35_1_2}) to an Abel equation we treat $\Omega$ as a new variable while $z$ we treat as a function of $\Omega$. Let us denote  $\Omega =: w$ and $z=Z(w)$. Hence, $1 = Z_{w} \Omega_{z}$ and (\ref{metryka_we_wspolrzednych_qpxw}) is proved.

Eq. (\ref{master_A_35_1_2}) reduces to the form
\begin{equation}
\label{master_A_35_1_2_wersja_2}
Z_{w} \left( Z -  \frac{ 2 \zeta_{0} \left( w + \frac{\Lambda}{3}  \right)}{12 w^{2} + \Lambda w} \right) + \frac{6w + 3 \Lambda}{12 w^{2} + \Lambda w} \, Z^{2} - \frac{4 \zeta_{0}}{12 w^{2} + \Lambda w} \, Z=0
\end{equation} 
Using the substitution (\ref{ogolne_podstawienie_Z}) with the definitions (\ref{Definicje_f0_f1_dla_algebray_A35_05}) one proves (\ref{master_A_35_1_2_wersja_3}). 
\end{proof}

\begin{Uwaga}
\normalfont
It is probably worth noting that the number of arbitrary constants in Eq.~(\ref{master_A_35_1_2_wersja_3}) can be reduced to one by the following substitution 
\begin{equation}
    w=\Lambda v,  \quad \zeta_0=\Lambda \beta_0, \quad \Sigma = \Lambda^{\frac{1}{2}} \widetilde{\Sigma}
\end{equation}
which eliminates $\Lambda$.
\end{Uwaga}

No solutions of Eq. (\ref{master_A_35_1_2_wersja_3}) with $\zeta_{0} \ne 0$ are known. The situation changes if $\zeta_{0} =0$. We analyze this case in Section \ref{Sekcja_5}.

\subsection{Algebra $A_{3,7}^{\alpha_{0}}$}
\label{Sekcja_Algebra_A_37}

For the algebra $A_{3,7}^{\alpha_{0}}$ we have $n_{0} = -1$, $a_{0}=1$ while $b_{0}=m_{0}=: \alpha_{0}$ where $\alpha_{0}$ is a constant such that $\alpha_{0} > 0$.

\begin{Twierdzenie}
\label{Twierdzenie_A_37}
Let $(\mathcal{M}, ds^{2})$ be an Einstein complex space of the type $[\textrm{D}]^{nn} \otimes [\textrm{I}]$ equipped with a 3D algebra of infinitesimal symmetries $A_{3,7}^{\alpha_{0}}$. Then there exists a local coordinate system $(q,p,x,y)$ such that the metric takes the form 
\begin{eqnarray}
\label{metryka_algebra_A37_wersja_3}
\frac{1}{2} ds^{2} &=& dqdy-dpdx + \frac{\Lambda}{3} \, \frac{s (3 +16 \alpha_{0}^{2} s)}{1 + G +  s + \alpha_{0}^{2} (s + 9G)} (xdq + y dp)^{2} 
\\ \nonumber
&& + \frac{\Lambda}{3} (1+s+G) (ydq - x dp)^{2}
\\ \nonumber
&& + \frac{2\Lambda \alpha_{0}}{3} \,  \frac{s (-1-4G-4s +12 \alpha_{0}^{2} (s-3G))}{1 + G +  s + \alpha_{0}^{2} (s + 9G)}  (y dq - xdp)(ydp + xdq)
\\ \nonumber
&& + \frac{\Lambda \alpha_{0}^{2}}{3} \,  \frac{-4s - 7s^{2} +9G + 9G^{2} +2sG + 9 \alpha_{0}^{2} (s - 3G )^{2}}{1 + G +  s + \alpha_{0}^{2} (s + 9G)}   (y dq - xdp)^{2}
\end{eqnarray}
where $\Lambda \ne 0$ is the cosmological constant, $\alpha_{0} >0$ is a constant, $s$ and $G$ read
\begin{equation}
\label{wspolrzedne_algebra_A37}
s =: \frac{1}{2^{\frac{4}{3}} r^{\frac{4}{3}}}, \  G =-\dfrac{ \dfrac{\Sigma}{2r} + \dfrac{1}{4} + \dfrac{1-3 \alpha_{0}^{2}}{2^{\frac{4}{3}} r^{\frac{4}{3}}}}{1 + 9 \alpha_{0}^{2}}
\end{equation}
where $\Sigma = \Sigma (r)$ is a holomorphic function which satisfies the equation
\begin{equation}
\label{rownanie_algebra_A_37_wersja_3}
\Sigma \Sigma_{r} - \Sigma = \left(  \frac{3}{4}r + 2^{\frac{2}{3}} r^{-\frac{1}{3}} \right) \left( 1 + \frac{\alpha_{0}^{2}}{3} 2^{\frac{8}{3}} r^{-\frac{4}{3}} \right)
\end{equation}
with
\begin{equation}
\label{zaleznosc_xy_i_r_A37}
(x^{2} + y^{2}) e^{2 \alpha_{0} \atan \left( \frac{x}{y} \right)}  = 
r^{\frac{2}{3}} \exp \left( - \int  \left( 1 +  \frac{2^{\frac{7}{2}} \alpha_{0}^{2}}{3r^{\frac{4}{3}} } \right) \frac{dr}{\Sigma}  \right)
\end{equation}
\end{Twierdzenie}

\begin{proof}
The constants $\zeta_{\dot{A}}$ can be gauged away. A general solution to the master equation reads
\begin{equation}
\Theta (x,y) = e^{- 4 \alpha_{0} \atan \left( \frac{x}{y} \right)  } H(z), \ \ z:= (x^{2} + y^{2}) \, e^{ 2 \alpha_{0} \atan \left( \frac{x}{y} \right)}
\end{equation}
This solution implies the following form of the second derivatives of the key function $\Theta$ 
\begin{eqnarray}
\label{Theta_ssecond_dderivatives}
\Theta_{yy} &=&  4 (y- \alpha_{0} x)^{2} H_{zz} +  \frac{2}{z}  (x^{2} + y^{2} + 6 \alpha_{0} x (y - \alpha_{0}x))H_{z} + 8 \alpha_{0} \frac{x}{z^{2}}  (2 \alpha_{0} x - y)H  \ \ \ \ 
\\ \nonumber
\Theta_{xx} &=&  4 (x+ \alpha_{0} y)^{2} H_{zz} +  \frac{2}{z}  (x^{2} + y^{2} - 6 \alpha_{0} y (x + \alpha_{0}y)) H_{z}+ 8 \alpha_{0} \frac{y}{z^{2}}  (2 \alpha_{0} y+x)H  \ \ \ \ 
\\ \nonumber
\Theta_{xy} &=& 4 (x+ \alpha_{0} y) (y- \alpha_{0} x)  H_{zz} +  \frac{6 \alpha_{0}}{z} (x^{2} - y^{2} + 2 \alpha_{0} xy )  H_{z} 
\\ \nonumber
&& + \frac{4 \alpha_{0}}{z^{2}}  (y^{2} - x^{2} - 4 \alpha_{0} xy)  H
\end{eqnarray}
The $\mathcal{HH}$-equation takes the form
\begin{eqnarray}
\label{HH_equation_for_algebra_A37}
(4-60 \alpha_{0}^{2}) H_{z}^{2} + 8 (1 + \alpha_{0}^{2}) z H_{z} H_{zz} + 32 \alpha_{0}^{2} H H_{zz} + 80 \alpha_{0}^{2} \frac{1}{z} H H_{z}
\\ \nonumber
-16 \alpha_{0}^{2} \frac{1}{z^{2}} H^{2} - \frac{\Lambda}{3} ( 4z^{2} H_{zz} - 4z H_{z} + 3H ) = 0
\end{eqnarray}
If we introduce the substitution
\begin{equation}
H (z) = - \frac{\Lambda}{3} \, T(v), \ v = z^{\frac{1}{2}}
\end{equation}
then Eq. (\ref{HH_equation_for_algebra_A37}) simplifies slightly and it reads
\begin{eqnarray}
\label{HH_equation_formalizm_T_algebra_A37}
&& 16 \alpha_{0}^{2} \left( \frac{1}{2v} TT_{vv} + \frac{2}{v^{2}} T T_{v} - \frac{1}{v^{3}} T^{2} - \frac{1}{v} T_{v}^{2} \right)
\\ \nonumber
&& \ \ \ \ \ \ \ + (1 + \alpha_{0}^{2} ) \, T_{v} T_{vv} + v^{3} T_{vv} - 3 v^{2} T_{v} + 3vT = 0
\end{eqnarray}
The metric expressed in terms of $T(v)$ yields
\begin{eqnarray}
\label{metryka_algebra_A37_wersja_2}
\frac{1}{2} ds^{2} &=& dqdy-dpdx + \frac{\Lambda}{3} \frac{T_{vv}}{v^{2}} (xdq + y dp)^{2} + \frac{\Lambda}{3} \left( \frac{T_{v}}{v^{3}} + 1 \right) (ydq - x dp)^{2}
\\ \nonumber
&& + \frac{2}{3} \alpha_{0} \Lambda \left( \frac{T_{vv}}{v^{2}} - \frac{4 T_{v}}{v^{3}} + \frac{4T}{v^{4}} \right) (y dq - xdp)(ydp + xdq)
\\ \nonumber
&& + \frac{\Lambda}{3} \alpha_{0}^{2} \left( \frac{T_{vv}}{v^{2}} - \frac{7 T_{v}}{v^{3}} + \frac{16T}{v^{4}} \right) (y dq - xdp)^{2}
\end{eqnarray}
Eq. (\ref{HH_equation_formalizm_T_algebra_A37}) can be reduced to a first--order ODE. We follow the pattern from Section \ref{Algebra_A34}. The substitution $T = v^{4} g (w )$, $w := \ln v$ gives 
\begin{equation}
\begin{split}
  3 g + 48 g^2 + 4 g_w  + 40 g g_w  + 7  g_w^2 +  g_{ww}  + 4 g g_{ww}  +  g_w g_{ww} \\
+ \alpha_{0}^{2} ( g_{w} g_{ww} + 12 g g_{ww} - 9 g_{w}^{2}  )= 0
\end{split}
\end{equation}
Then the substitution $g_w = Q-3g$, $g_{ww} = (Q_g-3)(Q-3g)$ leads to
\begin{equation}
\label{rownanie_algebra_A_37}
Q_g ( Q -3 g ) (1 + g +  Q + \alpha_{0}^{2} (Q + 9g)) + Q (1 + 4 g + 4 Q  - 12 \alpha_{0}^{2} (Q-3g)  ) = 0
\end{equation}
At this point one can calculate the discriminant $\mathcal{D}$ which turns out to be a complicated rational function of $g$ and $Q$
\begin{equation}
\label{wyroznik_dla_A37}
\mathcal{D}(g,Q)=\frac{64 \Lambda ^6 P_{3,7}}{27 \left(\alpha_0 ^2 (9
   g+Q)+g+Q+1\right)^{18}}    
\end{equation}
where
\begin{equation}
    P_{3,7}=27 \left(\alpha_0^2+1\right)^2 \left(9 \alpha_0 ^2+1\right)^4 Q^2
   (Q-2 g)^4 \left(2 \left(7 \alpha_0^2+3\right) Q-3 \left(9
   \alpha_0^2+1\right) g\right)^2+\dots
\end{equation}
is a polynomial of degree $22$ in $g$, $Q$. Derivatives of function $T(v)$ expressed in terms of $Q$ and $g$ read
\begin{equation}
\frac{T_{v}}{v^{3}} = Q + g, \ \frac{T_{vv}}{v^{2}} = \frac{(3 + 16 \alpha_{0}^{2} Q)Q}{1 + g + Q +\alpha_{0}^{2} (Q + 9g)}
\end{equation}
Now we treat $Q$ as a variable which we denote by $s$ and $g$ as a function of $s$, $g=G(s)$. Because $Q_{g} G_{s} = 1$ Eq. (\ref{rownanie_algebra_A_37}) takes the form
\begin{equation}
\label{rownanie_algebra_A_37_wersja_2}
 ( s -3 G ) (1 + G +  s + \alpha_{0}^{2} (s + 9G)) + s (1 + 4 G + 4 s  - 12 \alpha_{0}^{2} (s-3G)  ) \, \frac{dG}{ds} = 0
\end{equation}
and we arrive at the metric in the form (\ref{metryka_algebra_A37_wersja_3}). 
The relation between $x$, $y$ and $s$ is given by the formula
\begin{equation}
\label{relacja_miedzy_xy_i_s}
(x^{2} + y^{2}) e^{2 \alpha_{0} \atan \left( \frac{x}{y} \right)} = \exp \left( 2 \int \frac{G_{s}}{s - 3G} \, ds \right)
\end{equation}
The last step consists of the substitution (\ref{wspolrzedne_algebra_A37}) which brings Eq. (\ref{rownanie_algebra_A_37_wersja_2}) to the form (\ref{rownanie_algebra_A_37_wersja_3}) and the relation (\ref{relacja_miedzy_xy_i_s}) to the form (\ref{zaleznosc_xy_i_r_A37}).
\end{proof}

\begin{Uwaga}
\normalfont
Note that relation (\ref{zaleznosc_xy_i_r_A37}) cannot be solved for $x = x(y,r)$ or $y=y(x,r)$ in terms of elementary functions, regardless of its right-hand side. Hence, the metric cannot be explicitly written in terms of the coordinates $(q,p,x,r)$ or $(q,p,r,y)$ in this case.
\end{Uwaga}

\begin{Uwaga}
\normalfont
Note, that the substitution (\ref{wspolrzedne_algebra_A37}) makes sense if $1 + 9 \alpha_{0}^{2} \ne 0$. The case $1 + 9 \alpha_{0}^{2} = 0$ must be considered separately. However, within the real neutral case the constant $\alpha_{0}$ is real and the relation $1 + 9 \alpha_{0}^{2} = 0$ cannot be satisfied. Within the complex case we have $\alpha_{0} = \pm \frac{i}{3}$. However, the complex transformation of the algebra generators reduces this case to an algebra $A_{3,5}^{-\frac{1}{2}}$ which has been already analyzed in Section \ref{Algebra_A_35_05_Sekcja}.
\end{Uwaga}

Unfortunately, no solution to the Abel equation (\ref{rownanie_algebra_A_37_wersja_3}) is explicitly known. Nevertheless, the form of the discriminant (\ref{wyroznik_dla_A37}), although inconvenient, ensures existence of a corresponding nondegenerate space.

\subsection{Algebra $A_{3,6}$}

The case with $n_{0} = -1$, $a_{0}=1$, $b_{0}=m_{0}=0$ corresponds to an algebra $A_{3,6}$. It appears that the considerations are the same as in Section \ref{Sekcja_Algebra_A_37} with $\alpha_{0} = 0$.

\begin{Twierdzenie}
\label{Twierdzenie_A_36}
Let $(\mathcal{M}, ds^{2})$ be an Einstein complex space of the type $[\textrm{D}]^{nn} \otimes [\textrm{I}]$ equipped with a 3D algebra of infinitesimal symmetries $A_{3,6}$. Then there exists a local coordinate system $(q,p,x,y)$ such that the metric takes the form 
\begin{eqnarray}
\label{metryka_kanoniczna_A_36}
\frac{1}{2} ds^{2}&=& dq dy - dp dx +  \frac{\Lambda}{2^{\frac{4}{3}}} r^{-\frac{4}{3}} \left( \frac{3}{4} - \frac{\Sigma}{2r } \right)^{-1}  (y dp + x dq)^{2}
\\ \nonumber
 && \ \ \ \ \ \ \ \ \ \ \ \ \ \ \ 
+ \frac{\Lambda}{3} \left( \frac{3}{4} - \frac{\Sigma}{2r } \right) (y dq - x dp)^{2}
 \end{eqnarray}
where $\Lambda \ne 0$ is the cosmological constant, $\Sigma = \Sigma (r)$ is a holomorphic function which satisfies the equation
\begin{equation}
\label{HH_equation_for_A_36_formalizm_Sigma}
 \Sigma \Sigma_{r} - \Sigma = \frac{3}{4}r + 2^{\frac{2}{3}} r^{-\frac{1}{3}}
\end{equation}
and
\begin{equation}
\label{algebra_A_36_zwiazek_xy_r}
x^{2} + y^{2} =  r^{\frac{2}{3}} \exp \left( -  \int \frac{dr}{\Sigma} \right)
\end{equation}
\end{Twierdzenie}
\begin{proof}
It is enough to put $\alpha_{0}=0$ in the Theorem \ref{Twierdzenie_A_37}. Notice that this reduces equation (\ref{rownanie_algebra_A_37}) to the form (\ref{HH_equation_for_A_34_formalizm_Q}) found in the case of algebra $A_{3,4}$. It turns out that despite the different dependence of the key function on the variables $x$ and $y$, the discriminant $\mathcal{D}$ is also given by the formula (\ref{wyroznik_dla_algebry_A34_Qg}), with the same consequences on determining an algebraic type in the real neutral case. 
\end{proof}

\begin{Uwaga}
\normalfont
Unlike the case $A_{3,7}^{\alpha_{0}}$ the relations $x=x(y,r)$ or $y=y(x,r)$ can be explicitly obtained (compare (\ref{algebra_A_36_zwiazek_xy_r})). Hence, the metric can be explicitly written in terms of the coordinates $(q,p,x,r)$ or $(q,p,r,y)$. 
\end{Uwaga}

\begin{Uwaga}
\normalfont
Despite many similarities between algebras $A_{3,4}$ and $A_{3,6}$ the metrics in these cases are different. Note a subtle difference in the factors standing at $... \left( \frac{3}{4} - \frac{\Sigma}{2r } \right)^{-1}$ term. Not only the sign is different. In an algebra $A_{3,4}$ there is a factor $(y dq + x dp)^{2}$ while in an algebra $A_{3,6}$ the same factor reads $(y dp + x dq)^{2}$.
\end{Uwaga}


\section{Example}
\label{Sekcja_5}
\setcounter{equation}{0}

In this Section we focus on a space equipped with an algebra $A_{3,5}^{-\frac{1}{2}}$ with an additional assumption  $\zeta_{0} = 0$. From Theorem \ref{Twierdzenie_A_35_05} one finds that $\zeta_{0} = 0$ implies $f_{1} = f_{0} = 0$. Eq. (\ref{master_A_35_1_2_wersja_3}) can be immediately solved yielding $\Sigma  = \textrm{const}$. Finally, 
\begin{equation}
\label{rozwiazanie_A35_alpha_zero}
 Z(w) = z_{0} w^{-3} (12 w + \Lambda)^{\frac{5}{2}}
\end{equation}
where $z_{0}$ is a nonzero constant. The metric (\ref{metryka_we_wspolrzednych_qpxw}) can be written explicitly and it depends on two parameters $\Lambda \ne 0$ and $z_{0} \ne 0$
\begin{eqnarray}
\label{jedyna_jawna_metryka}
\frac{1}{2} ds^{2} &=& -\frac{z_0 (12 w + \Lambda)^{\frac{3}{2}}}{w^4 x^3}  \Big( 3 x (2w+\Lambda) \, dw + 2 w (12w+\Lambda) \, dx \Big) dq -dxdp
\\ \nonumber
&&  - x^{2} \left(\frac{3w - \Lambda}{2w+\Lambda} \right) \frac{\Lambda}{3} dp^{2} - \frac{2 z_0 }{w^3 x} \frac{(12w+\Lambda)^{\frac{7}{2}}}{(2w+\Lambda)} \frac{\Lambda}{3}dqdp
\\ \nonumber
&& + \frac{z_0^{2}}{3 w^6 x^{4}} \frac{(12w+\Lambda)^5 (36 w^2 +\Lambda^2)}{(2w+\Lambda)} dq^{2}
\end{eqnarray}

Since it is the only case we have managed to solve completely, it deserves more thorough analysis. The formulas for $\mathcal{D}$, $\mathcal{P}$ and $\mathcal{R}$ read
\begin{subequations}
\begin{eqnarray}
\mathcal{D} &=& \frac{2^{12}10^{6} \Lambda^{6} w^{12} (\Lambda + 12 w)^{2} (\Lambda^{2} -196 \Lambda w + 4 w^{2})  }{(2 w + \Lambda)^{16}}
\\
 \mathcal{P} &=& - \frac{51200 \Lambda  z_0^2 (\Lambda +12 w)^6 }{27 w^2 x^6 (\Lambda +2 w)^8}
 \\ \nonumber
   &&  \quad \times  \left(35 \Lambda^4-1512 \Lambda^2  w^2-1728 \Lambda  w^3+432 w^4+400 \Lambda^3 w\right)
   \\ 
   \mathcal{R}&=& \frac{163840000 \Lambda ^2 z_0^4 (\Lambda +12 w)^{12}}{729 w^6 x^{12}
   (\Lambda +2 w)^{15}}
    \left(\Lambda ^9-16298496  \Lambda ^2 w^7-32908032 \Lambda^3 w^6\right. 
  \\
  \nonumber
    && \quad -13473216 \Lambda ^4 w^5-2372256 \Lambda^5 w^4  -114720 \Lambda ^6 w^3-1456 \Lambda^7 w^2 
\\
\nonumber
    && \quad \left.+4665600 \Lambda  w^8+1119744   w^9-138 \Lambda^8 w\right)  
\end{eqnarray}
\end{subequations}
To find a range of $w$ in which the algebraic type of the ASD Weyl tensor satisfies criteria listed in Table \ref{kryteria_niezdegenerowania} we need to find roots and poles of $\mathcal{D}$, $\mathcal{P}$ and $\mathcal{R}$. We focus only on the points relevant for inequalities given in Table \ref{kryteria_niezdegenerowania}.

The roots of $\mathcal{D}$ can be easily found
\begin{equation}
\label{miejsca_zerowe_D}
    w_{\mathcal{D}_1} = \left( \frac{49}{2} + 10 \sqrt{6} \right) \Lambda \approx 48.9949\Lambda   \quad\quad w_{\mathcal{D}_2} = \left(  \frac{49}{2} - 10 \sqrt{6}  \right)\Lambda  \approx 0.0051\Lambda  
\end{equation}
Similarly, for $\mathcal{P}$ one finds\footnote{The exact algebraic form of roots  $w_{\mathcal{P}_1}$, $w_{\mathcal{P}_2}$, $w_{\mathcal{P}_3}$ and $w_{\mathcal{P}_4}$ is too obscure to be useful. Thus we work with numerical approximations.}
\begin{equation}
\label{miejsca_zerowe_P}
    w_{\mathcal{P}_1} \approx 4.7017 \Lambda \quad\quad w_{\mathcal{P}_2} \approx 0.2714 \Lambda \quad \quad  w_{\mathcal{P}_3} \approx -0.07033 \Lambda \quad \quad w_{\mathcal{P}_4} \approx -0.9028 \Lambda
\end{equation}
There are five real roots of $\mathcal{R}$ which can be found only numerically. There are also 4 complex roots which are not important for our considerations and one pole $w_{\mathcal{R}_4}$. These points read
\begin{multline}
\label{miejsca_zerowe_R}
    w_{\mathcal{R}_1} \approx 3.2990 \Lambda \quad\quad w_{\mathcal{R}_2} \approx 0.0065  \Lambda \quad \quad  w_{\mathcal{R}_3} \approx -0.0802 \Lambda \quad \quad w_{\mathcal{R}_4} = - 0.5 \Lambda   \\ w_{\mathcal{R}_5} \approx -1.1303  \Lambda \quad \quad w_{\mathcal{R}_6} \approx -5.8537 \Lambda \quad \quad 
\end{multline}

In a complex case it is clear that $\mathcal{D} \ne 0$ except points $w \in \{ w_{\mathcal{D}_1}, w_{\mathcal{D}_2}, 0, -\frac{\Lambda}{12} \}$ and except point $w = - \frac{\Lambda}{2}$ where it is undetermined. 

A discussion in a real case is slightly more complicated. Because of the factor $- \Lambda$ which stands in $\mathcal{P}$ the analysis must be separated into two cases, $\Lambda > 0$ and $\Lambda < 0$. 

Let us consider the case $\Lambda > 0$ first. Note, that the solution (\ref{rozwiazanie_A35_alpha_zero}) is real and nonzero only for $w > - \frac{\Lambda}{12}$. Denoting $w_{\mathcal{Z}} := - \frac{\Lambda}{12} \approx -0.0833 \Lambda$ we find
\begin{eqnarray}
\label{zakresy_dla_Lambda_wiecej}
\textrm{type } [\textrm{I}_{rc}]: \ && \ w \in ( w_{\mathcal{D}_{2}}, w_{\mathcal{D}_{1}})
\\ \nonumber
\textrm{type } [\textrm{I}_{c}]: \ && \ w \in  (w_{\mathcal{Z}}, w_{\mathcal{D}_{2}} ) \cup (w_{\mathcal{D}_{1}}, \infty) \setminus \{ 0 \}
\end{eqnarray}
If $\Lambda < 0$ then 
\begin{eqnarray}
\label{zakresy_dla_Lambda_mniej}
\textrm{type } [\textrm{I}_{r}]: \ && \ w \in ( w_{\mathcal{Z}}, w_{\mathcal{R}_{4}})
\\ \nonumber
\textrm{type } [\textrm{I}_{c}]: \ && \ w \in  (w_{\mathcal{R}_{4}}, \infty)
\end{eqnarray}
From (\ref{zakresy_dla_Lambda_wiecej}) and (\ref{zakresy_dla_Lambda_mniej}) it follows that the solution covers all algebraically general types in a real case.


\section{Concluding remarks}
\label{Section_5}
\setcounter{equation}{0}

In the current paper we furthered an analysis of the pKE-spaces which are algebraically general in both complex and real neutral cases. Our considerations were focused on spaces equipped with 3D algebra of infinitesimal symmetries containing a pair of commuting Killing vectors. With such a strong symmetry assumed, the $\mathcal{HH}$-equation was reduced to a first--order ODE which is an Abel equation of the second kind. All obtained equations are, unfortunately, quite nontrivial and for the most part we were not able to solve them. Nevertheless, in a single case we celebrated an explicit solution. The metric (\ref{jedyna_jawna_metryka}) is - according to our best knowledge - the first explicitly known example of a gpKE-metric.

This interesting result gives the ultimate answer to the question \textbf{Q} mentioned at the beginning of the paper. The answer is positive. The ASD Weyl tensor of an pKE-space can be of an arbitrary Petrov-Penrose type. Hence, a Cartan quartic of the associated $(2,3,5)$-distribution also has this property. 

It could be an interesting task to investigate gpKE-spaces equipped with semisimple 3D algebra of infinitesimal symmetries $A_{3,8}$ and $A_{3,9}$ (i.e., without an assumption that a pair of Killing vectors commute) or with 2D algebra of infinitesimal symmetries $2A_{1}$. In the latter case the $\mathcal{HH}$-equation has the form (\ref{rownanie_hiperniebianskie_dwa_komutujace_Killingi}). These problems are being studied now, but the results will be presented elsewhere.

\begin{thebibliography}{99}



\bibitem{Bor_Lamoneda_Nurowski} Bor G., Lamoneda L.H. and Nurowski P., \textsl{The dancing metric, G2-symmetry and projective rolling}, Trans. Am. Math. Soc. 370(6), 4433--4481 (2018)



\bibitem{Bor_Makhmali_Nurowski} Bor G., Makhmali O. and Nurowski P., \textsl{Para-K\"ahler-Einstein 4-manifolds and non-integrable twistor distributions}, Geometriae Dedicata 216, 9 (2022)




\bibitem{Cartan} Cartan E., \textsl{Les systèmes de Pfaff, à cinq variables et les équations aux dérivées partielles du second ordre}, Ann. Sci. École Norm. Sup. 3(27), 109--192 (1910)





\bibitem{Chudecki_klasyfikacja_Killingow_nieeks} Chudecki A., \textsl{Classification of the Killing vectors in nonexpanding $\mathcal{HH}$-spaces with $\Lambda$}, Classical and Quantum Gravity \textbf{29}, 135010


\bibitem{Chudecki_Ref_1} Chudecki A., \textsl{Hyperheavenly spaces and their application in Walker and para-K\"ahler geometries: Part I}, Journal of Geometry and Physics 179, 104591 (2022)

\bibitem{Chudecki_Ref_2} Chudecki A., \textsl{Hyperheavenly spaces and their application in Walker and para-K\"ahler geometries: Part II}, Journal of Geometry and Physics 188, 104826 (2023)



\bibitem{Chudecki_notes_on_congr} Chudecki A., \textsl{On geometry of congruences of null strings in 4-dimensional complex and real pseudo-Riemannian spaces}, Journal of Mathematical Physics \textbf{58}, 112502, (2017)

\bibitem{Chudecki_Examples} Chudecki A., \textsl{On some examples of para-Hermite and para-K\"{a}hler Einstein spaces with $\Lambda \ne 0$}, Journal of Geometry and Physics 112, 175--196 (2017)


\bibitem{Petrov} Petrov, A.Z., \textsl{Gravitational field geometry as the geometry of automorphisms},  proceedings of the conference Recent Developments in General Relativity, Pergamon Press (1962)


\bibitem{Plebanski_Finley_further} Finley J.D. III and Pleba\'nski J.F., \textsl{Further heavenly metrics and their symmetries}, Journal of Mathematical Physics, \textbf{17}, 585 (1976)


\bibitem{Finley_intrinsic} Finley J.D. III and Pleba\'nski J.F., \textsl{The intrinsic spinorial structure of hyperheavens}, Journal of Mathematical Physics, \textbf{17}, 2207 (1976)


\bibitem{Nurowski_An} Nurowski P. and D. An, \textsl{Twistor Space for Rolling Bodies}, Commun. Math. Phys. 326 393-414 (2013)



\bibitem{Panayotou} Panayotounakos D.E. and Zarmpoutis T.I., \textsl{Construction of Exact Parametric or Closed Form Solutions of Some Unsolvable Classes of Nonlinear ODEs (Abel's Nonlinear ODEs of the First Kind and Relative Degenerate Equations)}, International Journal of Mathematics and Mathematical Sciences, Volume 2011, 387429 (2011)




\bibitem{Patera} Patera J., Sharp R. T. and Winternitz P., \textsl{Invariants of real low dimension Lie algebras}, Journal of Mathematical Physics, Vol. 17, No.6 (1976)



\bibitem{Plebanski_spinors} Pleba\'nski J.F., \textsl{Spinors, tetrad and forms}, unpublished monography, CINVESTAV, M\'exico (1974)





\bibitem{Plebanski_Finley_Killing} Pleba\'nski J.F. and Finley J.D. III, \textsl{Killing vectors in nonexpanding HH spaces}, Journal of Mathematical Physics, \textbf{19}, 760 (1978)




\bibitem{Plebanski_Hacyan} Pleba\'nski J.F. and Hacyan S., \textsl{Null geodesic surfaces and Goldberg - Sachs theorem in complex Riemannian spaces}, Journal of Mathematical Physics, \textbf{16}, 2403 (1975)


\bibitem{Polyanin} Polyanin A.D. and Zaitsev V.F., \textsl{Handbook of Exact Solutions for Ordinary Differential Equations. Second Edition}, A CRC Press Company, 2003










\bibitem{Przanowski_Formanski_Chudecki} Przanowski M., Forma\'nski S. and Chudecki A., \textsl{Notes on para-Hermite-Einstein spacetimes}, International Journal of Geometric Methods in Modern Physics 9, 125008 (2012)





\bibitem{Walker} Walker A.G., \textsl{Canonical form for a Riemannian space with a parallel field of null planes},  Quart. J. Math. Oxford, (2) \textbf{1} 69 (1950)






\end{thebibliography}
\end{document}